\documentclass[10pt,conference]{IEEEtran}

\pdfoutput=1

\ifCLASSOPTIONcompsoc
  \usepackage[nocompress]{cite}
\else
  \usepackage{cite}
\fi

\ifCLASSINFOpdf
  \usepackage[pdftex]{graphicx}
\else
  \usepackage[dvips]{graphicx}
\fi

\usepackage{mathtools}
\usepackage{amsthm}
\usepackage{amssymb}
\usepackage{amsmath}
\usepackage{stmaryrd}

\usepackage{hyperref}
\hypersetup{
    colorlinks=true,
    citecolor=Green,
    linkcolor=NavyBlue,
    filecolor=magenta,
    urlcolor=NavyBlue,
  }

\usepackage[capitalise,noabbrev]{cleveref}

\usepackage[dvipsnames]{xcolor}
\usepackage{tikz-cd}
\usepackage{wrapfig}

\usepackage{enumerate}

\usepackage{fontawesome5}

\usepackage{microtype}

\DeclarePairedDelimiter{\set}{\{}{\}}
\DeclarePairedDelimiter{\pa}{(}{)}

\newcommand{\colonequiv}{\mathrel{\vcentcolon\mspace{-1.2mu}\equiv}}
\newcommand{\defeq}{\colonequiv}

\DeclareMathOperator{\inl}{\mathsf{inl}}
\DeclareMathOperator{\inr}{\mathsf{inr}}

\DeclarePairedDelimiter{\squash}{\|}{\|}

\newcommand{\proptrunc}{\squash}

\newcommand{\Zero}{\mathbf 0}
\newcommand{\One}{\mathbf 1}
\newcommand{\Two}{\mathbf 2}

\newcommand{\lambdadot}[2]{\mathop{\lambda}{#1}\mathrel{.}#2}

\renewcommand{\iff}{\ensuremath{\leftrightarrow}}

\newcommand{\U}{\mathcal U}

\DeclareMathOperator{\acc}{\mathsf{is-accessible}}

\DeclareMathOperator{\Prop}{\mathsf{Prop}}

\newcommand{\V}{\mathbb V}
\newcommand{\Vord}{\V_{\textup{ord}}}
\DeclareMathOperator{\Vsetop}{\V-\textup{\textsf{set}}\hspace{1pt}}
\newcommand{\Vset}[2]{\Vsetop\pa*{#1,#2}}

\newbox\qqBoxA
\newdimen\qqCornerHgt
\setbox\qqBoxA=\hbox{$\ulcorner$}
\global\qqCornerHgt=\ht\qqBoxA
\newdimen\qqArgHgt
\def\Quinequote #1{%
    \setbox\qqBoxA=\hbox{$#1$}%
    \qqArgHgt=\ht\qqBoxA%
    \ifnum     \qqArgHgt<\qqCornerHgt \qqArgHgt=0pt%
    \else \advance \qqArgHgt by -\qqCornerHgt%
    \fi \raise\qqArgHgt\hbox{$\ulcorner$} \box\qqBoxA %
    \raise\qqArgHgt\hbox{$\urcorner$}}

\newcommand{\Ord}{\mathsf{Ord}}

\newcommand{\mewo}{\ensuremath{\mathsf{MEWO}}}
\newcommand{\cewo}{\ensuremath{{\mathsf{MEWO}_{\mathsf{cov}}}}}

\DeclareMathOperator{\mar}{m}
\DeclareMathOperator{\Mar}{M}
\newcommand{\simbelow}{\mathrel{\leq_{\Mar}}}

\newcommand{\presinitial}[2]{\mathsf{psim}_{#1}(#2)}

\newcommand{\initseg}{\mathbin\downarrow}
\newcommand{\downs}[1]{\mathop{\downarrow^{#1}}}
\newcommand{\down}{\downs{+}}

\DeclareMathOperator{\totalspace}{\mathbb T}
\DeclareMathOperator{\im}{im}
\newcommand{\insmall}{\mathbin\in_{\U}}

\newcommand{\tc}{\mathrel{<^+}}
\newcommand{\trc}{\mathrel{<^*}}

\newcommand{\sing}[1]{\ensuremath{\{#1\}}}

\newcommand{\qedNoProof}{\qed}

\newcommand{\formalised}{{\color{NavyBlue!75!White}{\raisebox{-0.5pt}{\scalebox{0.8}{\faCog}}}}}
\newcommand{\flinkurl}[1]{\href{#1}{\formalised}}
\newcommand{\flink}[1]{\flinkurl{https://tdejong.com/agda-html/st-tt-ordinals/index.html\##1}}

\begin{document}

\title{Set-Theoretic and Type-Theoretic Ordinals Coincide}

\author{\IEEEauthorblockN{Tom de Jong\IEEEauthorrefmark{1},
  Nicolai Kraus\IEEEauthorrefmark{1},
  Fredrik Nordvall Forsberg\IEEEauthorrefmark{2} and
  Chuangjie Xu\IEEEauthorrefmark{3}}
\IEEEauthorblockA{\IEEEauthorrefmark{1}School of Computer Science,
  University of Nottingham,
  Nottingham, UK\\
  Email: \{tom.dejong, nicolai.kraus\}@nottingham.ac.uk}
\IEEEauthorblockA{\IEEEauthorrefmark{2}Department of Computer and Information Sciences,
  University of Strathclyde,
  Glasgow, UK\\
  Email: fredrik.nordvall-forsberg@strath.ac.uk}
\IEEEauthorblockA{\IEEEauthorrefmark{3}Research and Development Team,
  SonarSource GmbH,
  Bochum, Germany\\
  Email: chuangjie.xu@sonarsource.com}
}

\theoremstyle{plain}
\newtheorem{theorem}{Theorem}
\newtheorem{lemma}[theorem]{Lemma}
\newtheorem{proposition}[theorem]{Proposition}
\newtheorem{corollary}[theorem]{Corollary}
\theoremstyle{definition}
\newtheorem{definition}[theorem]{Definition}
\newtheorem{example}[theorem]{Example}
\theoremstyle{remark}
\newtheorem{remark}[theorem]{Remark}

\thispagestyle{plain}
\pagestyle{plain}

\IEEEoverridecommandlockouts
\IEEEpubid{\phantom{\makebox[\columnwidth]{979-8-3503-3587-3/23/\$31.00~
\copyright2023 IEEE \hfill} \hspace{\columnsep}\makebox[\columnwidth]{ }}}

\maketitle

\begin{abstract}
  In constructive set theory, an ordinal is a hereditarily transitive set.  In
  homotopy type theory (HoTT), an ordinal is a type with a transitive,
  wellfounded, and extensional binary relation. We show that the two definitions
  are equivalent if we use (the HoTT refinement of) Aczel's interpretation of
  constructive set theory into type theory.
  Following this, we generalize the notion of a type-theoretic ordinal to
  capture \emph{all} sets in Aczel's interpretation rather than only the
  ordinals.  This leads to a natural class of ordered structures which contains
  the type-theoretic ordinals and realizes the higher inductive interpretation
  of set theory.
  All our results are formalized in Agda.
\end{abstract}

\section{Introduction}

Set theory and dependent type theory are two very different settings in which
constructive mathematics can be developed, but not always in comparable ways.
Lively discussions on what foundation is ``better'' are not
uncommon.
While we do not dare to offer a judgment on this question, we can at least
report that the choice of foundation is in a certain sense insignificant for the
development of constructive ordinal theory.  We consider this an interesting
finding since ordinals are fundamental in the foundations of set theory and are
used in theoretical computer science in termination arguments~\cite{Floyd1967}
and semantics of inductive definitions~\cite{Aczel1977,DybjerSetzer1999}.

In constructive set theory, following Powell's seminal work~\cite{Powell1975}, a
standard definition%
\footnote{In a classical setting, Cantor's ordinals~\cite{Cantor1883} can be
  presented in multiple equivalent ways. In a constructive setting, these
  presentations are not equivalent and are often not as well-behaved as one
  would wish.  Therefore, various reasonable definitions of ordinals
  are known and have been studied; for example, Taylor~\cite{Taylor1996} gave a
  constructively better-behaved formulation, and in previous work, we compared
  several approaches~\cite{KrausNordvallForsbergXu2022}.  }
of an ordinal is that of a \emph{transitive} set whose elements are again
transitive sets (also cf.~Aczel and Rathjen~\cite{AczelRathjen2010}). A set
\(x\) is transitive if for every \(y \in x\) and \(z \in y\), we have
\(z \in x\), i.e., if \(y \in x\) implies \(y \subseteq x\).
Note how this definition makes essential use of how the membership predicate
\({\in}\) in set theory is global, by simultaneously referring to \(z \in y\)
and \(z \in x\).
In type theory, on the other hand, the statement ``if \(y : x\) and \(z : y\)
then \(z : x\)'' is ill-formed, and so ordinals need to be defined differently.
In homotopy type theory, an ordinal is defined to be a type equipped with an
order relation that is transitive, extensional, and
wellfounded~\cite[\S{}10.3]{HoTTBook}.

A priori, the set-theoretic and the type-theoretic approaches to ordinals are
thus quite different.  One way to compare them is to interpret one foundation
into the other.  Aczel~\cite{Aczel1978} gave an interpretation of Constructive
ZF set theory into type theory using so-called setoids, which was later refined
using a higher inductive type \(\V\) in the HoTT book~\cite[\S{}10.5]{HoTTBook},
referred to as the \emph{cumulative hierarchy}.  Through this construction,
homotopy type theory hosts a model of set theory, and we make use of this to
study the set-theoretic approach to ordinals within it.

To be specific, the cumulative hierarchy \(\V\) allows us to define a set
membership relation \({\in}\), which makes it possible to consider the type
\(\Vord\) of elements of \(\V\) that are set-theoretic ordinals.  Similarly, we
write $\Ord$ for the type of all type-theoretic ordinals, i.e., for the type of
transitive, extensional, and wellfounded order relations.  We show that $\Vord$ and
$\Ord$ are equivalent (isomorphic), meaning that we can translate between
type-theoretic and set-theoretic ordinals.

This translation by itself would not be satisfactory if it were not well-behaved;
what makes it valuable is that it preserves the respective order.  A fundamental
result of type-theoretic ordinals is that the type \(\Ord\) of (small)
ordinals is itself a type-theoretic ordinal when ordered by inclusion of
strictly smaller initial segments (also referred to as \emph{bounded
  simulations}).  To complement this, we show that the type $\Vord$ of
set-theoretic ordinals also canonically carries the structure of a
type-theoretic ordinal.  The isomorphisms that we construct respect these
orderings, and our first main result (\cref{Ord-and-Vord-coincide}) is that
$\Ord$ and $\Vord$ are isomorphic as ordinals (and, consequently, equal, by a
standard application of univalence).  Thus, the set-theoretic and type-theoretic
approaches to ordinals coincide in homotopy type theory.

Going further, we dive deeper into the study of the isomorphism
\(\Vord \to \Ord\).  The analogue to this function in set theory computes the
\emph{rank}~\cite{firstrank,Aczel1977,AczelRathjen2010} of sets recursively.
While our definition is recursive as well, we show that it is possible to give a
conceptually simpler, non-recursive description of the rank of transitive sets,
although this requires paying close attention to size issues.
Specifically, we show that the rank of a set-theoretic ordinal \(\alpha\) is
isomorphic to --- but not equal to for size reasons --- the type of all members
of~\(\alpha\) (\cref{cor:rank-nonrecursively}).

In the second part of the paper, we generalize the isomorphism between set- and
type-theoretic ordinals.
Given that the subtype \(\Vord\) of \(\V\) is isomorphic to \(\Ord\), a type of
ordered structures, it is natural to ask what type of ordered structures
captures \emph{all} of \(\V\).
\begin{wrapfigure}{r}{0.2\textwidth}
  \begin{equation}\label{big-picture}
    \begin{tikzcd}
      \Vord \ar[d,hookrightarrow] \ar[r,"\simeq"] & \Ord \ar[d,hookrightarrow] \\
      \V \ar[r,"\simeq"] & T
    \end{tikzcd}
  \end{equation}
\end{wrapfigure}
That is, we look for a natural type \(T\) of ordered
structures such that the diagram on the right commutes.
Since \(\V\) is \(\Vord\) with transitivity dropped, it is tempting to
try to choose $T$ to be \(\Ord\) without transitivity, i.e., the type
of extensional and wellfounded relations. However, such an attempt is too
naive to work:
consider the type-theoretic ordinal \(\alpha\) with two elements \(0 < 1\), whose
corresponding set in \(\Vord\) is the set \(2 = \{\emptyset, \{\emptyset\}\}\).
The latter is the set-theoretic transitive closure of the non-transitive set
\(\{\{\emptyset\}\} \subseteq 2\), but the only extensional, wellfounded order
whose order-theoretic transitive closure is \(\alpha\) is \(\alpha\) itself.
In other words, there cannot be an order-preserving isomorphism between \(\V\)
and the type of extensional, wellfounded order relations, since there is no corresponding
order for the set \(\{\{\emptyset\}\}\) --- we need additional structure to
fully capture this set. 

To this end, we introduce the theory of \emph{(covered) marked extensional
  wellfounded orders (mewos)}, i.e., extensional, wellfounded relations with
additional structure in the form of a \emph{marking}.
The idea is that the carrier of the order also contains elements representing
elements of elements of the set, with the marking designating the ``top-level''
elements: the set \(\{\{\emptyset\}\}\) is again represented by the order \(\alpha\)
with two elements \(0 < 1\), but with only element \(1\) marked.
Such a marking is \emph{covering} if any element can be reached from a marked
top-level element, i.e., if the order contains no ``junk''.
Since every ordinal can be equipped with the trivial covering by marking all
elements, the type \(\Ord\) of ordinals is a subtype of the type \(\cewo\) of covered
mewos, as requested by Diagram~\ref{big-picture}.

The idea of encoding sets as wellfounded structures is not new; see, e.g.,
\cite[\S{}7]{Osius1974}, \cite[\S{}3]{Taylor1996},
\cite[\S{}4.7]{AdamekMiliusMossSousa2013} and Aczel's~\cite{Aczel1988} ``canonical
picture''~\cite[Ex~4.22]{AdamekMiliusMossSousa2013}.
Instead, the point is to have a notion that allows for a smooth type-theoretic
generalization of the theory of ordinals.
Additionally, covered mewos are shown to work predicatively (i.e., without the
need to assume resizing axioms), which is not obvious for the previously
mentioned approaches.

Aiming for an isomorphism \(\V \simeq \cewo\), we develop the theory of the
covered mewos: the type of covered mewos is itself a covered mewo, and it has
both 
a successor operation, and least upper bounds of arbitrary
(small) families of covered mewos.
Compared to the theory of ordinals, some additional care is required as the
orders involved are not assumed to be transitive.
Using successors and least upper bounds of mewos, we construct a map
\(\V \to \cewo\) by the recursive formula for the rank of a set, and show that
it has an inverse.

\subsection{Summary of contributions}

\begin{itemize}
\item We show that set-theoretic and type-theoretic ordinals coincide
  (\cref{Ord-and-Vord-coincide}).
\item We show that the rank of an ordinal can be defined in a non-recursive way
  (\cref{cor:rank-nonrecursively}).
\item We show that the model of set theory $\V$ is equivalently represented by
  the structure of covered marked extensional wellfounded order relations
  (\cref{thm:second-main-result}).
\end{itemize}

\subsection{Related work}

Constructive treatments of ordinals can be found in Joyal and
Moerdijk~\cite{JoyalMoerdijk1995} and Taylor~\cite{Taylor1996}.
In the context of homotopy type theory, what we call type-theoretic ordinals
were developed in the HoTT book~\cite[\S{}10.3]{HoTTBook}, and their theory
significantly expanded by Escard\'o and
collaborators~\cite{TypeTopologyOrdinals}.
In previous work~\cite{KrausNordvallForsbergXu2021ordHoTT,KrausNordvallForsbergXu2022},
we
developed a framework for different notions of constructive ordinals, and showed
that all ordinals we considered embed into the type-theoretic ordinals in an
order-preserving way.
In fact, the current paper grew out of an attempt to locate the set-theoretic
ordinals somewhere between the countable Brouwer tree ordinals (as considered
by, e.g., Brouwer~\cite{Brouwer1927}, Church~\cite{Church1938},
Kleene~\cite{Kleene1938}, Martin-L\"of~\cite{MartinLof1970}, and Coquand,
Lombardi and Neuwirth~\cite{Coquand2022}) and the type-theoretic ordinals in
this framework, before we realised that they actually coincide with the latter!

To the best of our knowledge, the first interpretation of Constructive ZF set
theory into type theory was given by Aczel~\cite{Aczel1978}.  This original
interpretation uses so-called setoids and has a form of choice built-in.  It was
later refined using a higher inductive type \(\V\) in the HoTT
book~\cite[\S{}10.5]{HoTTBook}, referred to as the \emph{cumulative hierarchy},
and this is the version we use in the current paper.
Gylterud~\cite{Gylterud2018} showed that \(\V\) can be constructed using only an
ordinary inductive type without higher constructors. Although it is not our main
motivation, the current paper demonstrates that \(\V\) can be realized not as an
inductive type at all, but as the collection of all covered marked wellfounded
extensional relations (however, the notion of wellfoundedness is defined as an
inductive type, and the notion of coveredness uses higher constructors in the
form of propositional truncations).
Taylor~\cite{Taylor1996} also considers wellfounded extensional relations (which he
calls \emph{ensembles}) as ``codes'' for sets in an elementary topos, but does
not consider markings on them.
Coverings and markings are what allow us to achieve completeness, i.e., to
represent \emph{all} sets in \(\V\).

\subsection{Setting, assumptions, and notation}

We work in and assume basic familiarity with homotopy type theory as introduced
in the HoTT book~\cite{HoTTBook}, i.e., Martin-L\"of type theory extended with
higher inductive types and the univalence axiom.
We also follow this book closely regarding notation and denote the Martin-L\"of
identity type by \(a = b\), while \(a \equiv b\) is reserved for definitional
(also referred to as judgmental) equality.
Universe levels are kept implicit, and we write \(\U^+\) for the next universe
containing the universe \(\U\).  For an implicitly fixed universe $\U$, we write
$\Prop$ or $\Prop_\U$ for the subtype of propositions,
$\Prop \colonequiv \Sigma(P : \U). \mathsf{is}\mbox{-}\mathsf{prop}(X)$, where a
\emph{proposition} is a type with at most one element
(``proof-irrelevant''). Following standard terminology, a \emph{set} is a type
whose identity types are propositions.

We write the type of dependent functions \(\Pi(x : A). B(x)\) as
\(\forall(x : A). B(x)\) when $B(x)$ is known to be a family of propositions.
Moreover, we denote by
${\exists (x : A). B(x)}$ the propositional truncation of the type of dependent
pairs \(\proptrunc{\Sigma(x : A). B(x)}\).

\subsection{Formalization}

All our results have been formalized in the Agda proof assistant, and type checks using Agda 2.6.3.
Our formalization of \cref{sec:ordinals} is building on Escard\'o's TypeTypology
library~\cite{type-topology}, whereas our formalization of \cref{sec:mewo} is
building on the agda/cubical library~\cite{agdacubical}.
The formalization has been archived with the DOI
\href{https://doi.org/10.5281/zenodo.7857275}{\nolinkurl{10.5281/zenodo.7857275}}, and an HTML rendering of our Agda code is also available at \url{https://tdejong.com/agda-html/st-tt-ordinals/}.
Throughout (the arXiv version of) our paper, the symbol \formalised{} is a clickable link to the
corresponding machine-checked statement.

\section{Ordinals in type theory and set theory}\label{sec:ordinals}
We start by reviewing both the set-theoretic and type-theoretic approaches to
ordinals.
We then recall the higher inductive construction of a model of constructive set
theory in homotopy type theory~\cite[\S{}10.5]{HoTTBook},
allowing us to consider the set-theoretic ordinals inside homotopy type theory,
and to prove that they coincide with the type-theoretic ordinals.
Finally, we revisit a recursive aspect of our proof and provide alternative
non-recursive constructions, which require paying close attention to type
universe levels.

\subsection{Ordinals in homotopy type theory}

The theory of ordinals in homotopy type theory was introduced in the HoTT
book~\cite[\S{}10.3]{HoTTBook} and significantly expanded on by Escard\'o and
collaborators~\cite{TypeTopologyOrdinals}.  One of the core concepts is
wellfoundedness which, constructively, is conveniently phrased in terms of
accessibility:

\begin{definition}[\flink{Definition-1} Accessibility]
  For a type \(X\) equipped with a binary relation \(<\), the
  type family \(\acc_{<}\) on \(X\) is inductively defined by saying that
  \(\acc_<(x)\) holds if \(\acc_<(y)\) holds for every \(y < x\).
\end{definition}

The point of accessibility is that it captures the principle of transfinite
induction by a \emph{single} inductive definition.

\begin{lemma}[\flink{Lemma-2} Transfinite induction]
  For a type \(X\) equipped with a binary relation \(<\), every element of \(X\)
  is accessible if and only if for every type family \(P\) on \(X\), we have
  \(P(x)\) for all \(x : X\) as soon as for every \(x : X\), the statement
  \(\forall (y : X).{y < x \to P(y)}\) implies \(P(x)\).
  \qedNoProof
\end{lemma}

Cantor's original definition of ordinal numbers was that of isomorphism
classes of well-ordered sets, but using univalence, all
representatives of a given isomorphism class in homotopy type theory
are identical.
Hence, we can use the well-ordered sets directly to represent ordinals.
The classical definition of well-order states that every non-empty subset
has a minimal element.
Constructively, the following (classically equivalent) formulation is better behaved.

\begin{definition}[\flink{Definition-3} Type-theoretic ordinal]
  A binary relation \(<\) on a type \(X\) is said to be
  \begin{enumerate}[(i)]
  \item\emph{prop-valued} if \(x < y\) is a proposition for every \(x,y : X\);
  \item\emph{wellfounded} if every element of \(X\) is accessible with respect
    to \(<\), i.e., \(\forall(x : X).\acc_<(x)\);
  \item\emph{extensional} if \(\forall(z : X).\pa*{z < x \iff z < y}\) implies
    \(x = y\) for every \(x,y : X\); and
  \item\emph{transitive} if \(x < y\) and \(y < z\) together imply \(x < z\) for
    every \(x,y,z : X\).
  \end{enumerate}
  A \emph{(type-theoretic) ordinal} is a type \(X\) with a binary relation \(<\)
  on \(X\) that is prop-valued, wellfounded, extensional, and transitive.
\end{definition}

\begin{remark}
  While~\cite{HoTTBook} requires the carrier of an ordinal to be a set (in the
  sense of HoTT), Escard\'o~\cite{TypeTopologyOrdinals} observed that this
  follows from prop-valuedness and extensionality.
\end{remark}

We now recall the notion of an initial segment and bounded simulation, which
will play fundamental roles in our constructions and proofs.
\begin{definition}[\flink{Definition-5} Initial segment,
  \(\alpha \initseg a\); bounded simulation, \(<\)]\label{def:initial-segment}
  An element \(a\) of an ordinal \(\alpha\) determines an \emph{initial segment}
  of \(\alpha\) defined as
  \[
    \alpha \initseg a \;\colonequiv\; \Sigma (x : \alpha).{x < a},
  \]
  which is again an ordinal with the order induced by \(\alpha\).  A
  \emph{bounded simulation} between ordinals, \(p : \alpha < \beta\), is a proof
  that $\alpha$ is an initial segment of $\beta$,
  \[
    \alpha < \beta \; \colonequiv \; %
    (\Sigma (b : \beta).{\alpha \simeq \beta \initseg b}).
  \]
\end{definition}

Note that the definition of \(\alpha < \beta\) above is equivalent to the
definition given in \cite[Def~10.3.19]{HoTTBook}; in particular, it is a
proposition.

\begin{remark}\label{rem:equivalence-instead-of-equality}
  In \cref{def:initial-segment} above, we could have defined a bounded
  simulation using an identification \({\alpha = \beta \initseg b}\), but opted
  for an equivalence \({\alpha \simeq \beta \initseg b}\) instead.  These two
  expressions are equivalent by univalence. However, the latter has the
  advantage of begin \emph{small}, i.e., living in the same universe as $\alpha$
  and $\beta$, while the former lives in the next universe.
\end{remark}

\begin{theorem}[\flink{Theorem-7}]
  The type \(\Ord\) of ordinals in a univalent universe, together with the
  relation $<$ of bounded simulations, is itself a (large) type-theoretic
  ordinal. \qedNoProof
\end{theorem}

A bounded simulation is a special case of the following more general definition
that serves as a notion of morphism between ordinals:
\begin{definition}[\flink{Definition-8} Simulation, \({\leq}\)]
  A \emph{simulation} between two ordinals \(\alpha\) and \(\beta\) is a
  function \(f\) between the underlying types satisfying:
  \begin{enumerate}[(i)]
  \item \emph{monotonicity:} \(x <_\alpha y\) implies \(f \, x <_\beta f \, y\)
    for every two elements \(x,y : \alpha\), and
  \item \emph{the initial segment property:} for every \(x : \alpha\) and
    \(y : \beta\), if \(y <_\beta f \, x\), then there is a \(x' <_\alpha x\)
    with \(f \, x' = y\).
  \end{enumerate}
  If we have a simulation between \(\alpha\) and \(\beta\), then, motivated by
  \cref{simulation-properties} below, we denote this by \(\alpha \leq \beta\).
\end{definition}

We stress that $\leq$ is a primitive relation, and \emph{not} given as the
disjunction of $<$ and equality --- in fact, we have
${\alpha \leq \beta} \leftrightarrow \big((\alpha < \beta) + (\alpha =
\beta)\big)$ for all ordinals $\alpha$ and $\beta$ if and only if the law of
excluded middle holds~\cite[Thm~64]{KrausNordvallForsbergXu2022}.

\begin{proposition}[\flink{Proposition-9}]
  \label{simulation-properties}
  Simulations make \(\Ord\) into a poset. Moreover, for ordinals
  \(\alpha\) and \(\beta\), the following are equivalent:
  \begin{enumerate}[(i)]
  \item \(\alpha \leq \beta\),
  \item for every \(\gamma : \Ord\), if \(\gamma < \alpha\), then
    \(\gamma < \beta\), and
  \item\label{item:prop9-three} for every \(a : A\), we have a (necessarily
    unique) \(b : \beta\) with \(\alpha \initseg a = \beta \initseg b\).
  \end{enumerate}
  Given \(f : \alpha \leq \beta\), the element \(b\) in \eqref{item:prop9-three}
  is given by \(f(a)\). \qedNoProof
\end{proposition}

The relation $\leq$ on $\Ord$ is antisymmetric, which can be used to prove that
two ordinals are isomorphic; however, it is often convenient to work with the
following alternative description.

\begin{lemma}[\flink{Lemma-10}]\label{iso-of-ordinals}
  A map between ordinals is an isomorphism if and only if it is bijective and
  preserves and reflects the order. \qedNoProof
\end{lemma}

Univalence implies that isomorphic ordinals in the same universe are equal (as
ordinals), which allows us to prove the equalities in the upcoming lemmas.

\begin{lemma}[\flink{Lemma-11}]\label{iterated-initial-segments}
  For \(p : a < b\) in an ordinal $\alpha$, iterations of initial segments
  simplify as follows:
  $\pa*{\alpha \initseg a} \initseg (b,p) = \alpha \initseg b$.  \qedNoProof
\end{lemma}

Besides initial segments, we will need two additional constructions of ordinals,
sums and suprema, as well as a few lemmas expressing how these interact with
initial segments.

\begin{definition}[\flink{Definition-12} Sum of ordinals, \(\alpha + \beta\)]%
  \label{sum-of-ordinals}
  Given two ordinals \(\alpha\)~and~\(\beta\), we construct another ordinal, the
  \emph{sum} \(\alpha + \beta\), by ordering the coproduct of the underlying
  types of \(\alpha\) and \(\beta\) as
  \begin{align*}
    \inl a < \inr b &\colonequiv \One,
    &\inl a < \inl a' &\colonequiv a <_\alpha a', \\
    \inr b < \inl a &\colonequiv \Zero,
    &\inr b < \inr b' &\colonequiv b <_\beta b'.
  \end{align*}
\end{definition}

Initial segments of sums obey the following laws:

\begin{lemma}[\flink{Lemma-13}]\label{initial-segments-of-sums}
  For ordinals \(\alpha\) and \(\beta\), and \(a : \alpha\), we have:
  \begin{enumerate}[(i)]
  \item \({(\alpha + \beta)} \initseg {\inl a} = \alpha \initseg a\), and
  \item \({(\alpha + \One)} \initseg {\inr \star} = \alpha\). \qedNoProof
  \end{enumerate}
\end{lemma}

\begin{definition}[\flink{Definition-14} Supremum of ordinals,
  \(\bigvee_{i : I}\alpha_i\)]\label{sup-of-ordinals}
  Given a type \(I : \U\) and a family of \(\alpha : I \to \Ord\) of ordinals in
  \(\U\), we construct another ordinal, the \emph{supremum}
  \(\bigvee_{i : I}\alpha_i\), as the set quotient of \(\Sigma (i : I). \alpha_i\)
  by the relation
  \[
    (i,x) \approx (j,y) \;\colonequiv\; \pa*{\alpha_i \initseg x \simeq \alpha_j \initseg y}
  \]
  and ordered by
  \[
    [i,x] < [j,y] \;\colonequiv\; \pa*{\alpha_i \initseg x < \alpha_j \initseg y}.
  \]
\end{definition}

Note that the distinction between $\simeq$ and $=$ discussed in
\cref{rem:equivalence-instead-of-equality} is important in the definition above.
It ensures that the supremum \(\bigvee_{i : I}\alpha_i\) lives in the
``correct'' universe, i.e., is an element of $\Ord$.

The name ``supremum'' comes from the fact that \(\bigvee_{i : I}\alpha_i\)
indeed is the supremum (least upper bound) of the family
\({\alpha : I \to \Ord}\) in the poset \(\Ord\), as shown
in~\cite[Thm~5.8]{deJongEscardo2022} which extends~\cite[Lem~10.3.22]{HoTTBook}.
In particular, we have simulations \(\alpha_j \leq \bigvee_{i : I}\alpha_i\) for
every \(j : I\) given by \(x \mapsto [j,x]\).

\begin{lemma}[\flink{Lemma-15}]\label{initial-segments-of-sups}
  Initial segments of suprema obey the following laws for all families
  \(\alpha : I \to \Ord\) of ordinals:
  \begin{enumerate}[(i)]
  \item \(\bigvee_{i : I}\alpha_i \initseg [j,x] = \alpha_j \initseg x\) for
    all \(j : I\) and \(x : \alpha_j\), and
  \item for every \(y : \bigvee_{i : I}\alpha_i\), there exist \(j : J\) and
    \(x : \alpha_i\) for which
    \(\bigvee_{i : I}\alpha_i \initseg y = \alpha_j \initseg x\).
  \end{enumerate}
  Thus, an initial segment of a supremum is given by an initial segment of a
  component.
\end{lemma}
\begin{proof}
  The first property follows from~\cref{simulation-properties} and the fact that
  for every \(j : I\), the map \(x \mapsto [j,x]\) from \(\alpha_j\) to
  \(\bigvee_{i : I}\alpha_i\) is a simulation.
  The second follows from the first and the surjectivity of the map
  \([-] : \pa[Big]{\Sigma(j : J). \alpha_j} \to \bigvee_{j : J}\alpha_j\).
\end{proof}

\subsection{Ordinals in set theory}\label{sec:ordinals-in-set-theory}
In constructive set theory, following Powell~\cite{Powell1975}, the standard
definition~\cite[Def~9.4.1]{AczelRathjen2010} of an ordinal is simple to
state: it is a transitive set whose elements are again transitive sets.

\begin{definition}[\flink{Definition-16} Transitive set]
  A set \(x\) is \emph{transitive} if for every \(z\) and \(y\) with \(z \in y\)
  and \(y \in x\), we have \(z \in x\).
\end{definition}

Note how this definition makes essential use of how the membership
predicate \({\in}\) in set theory is global, by simultaneously
referring to \(z \in y\) and \(z \in x\).

\begin{example}\label{ex:transitive-sets}
  The sets \(\emptyset\), \(\set{\emptyset}\),
  \(\set{\emptyset,\set{\emptyset}}\) and
  \(\set{\emptyset,\set{\emptyset},\set{\set{\emptyset}}}\) are all transitive, but
  \(\set{\set{\emptyset}}\) is not, because \(\emptyset\) is not a member.
\end{example}

\begin{definition}[\flink{Definition-18} Set-theoretic ordinal]\label{set-theoretic-ordinal}
  A \emph{set-theoretic {ordinal}} is a hereditarily transitive set, i.e., a
  transitive set whose elements are all transitive sets.
\end{definition}

The first three sets of~\cref{ex:transitive-sets} are all ordinals, but the
fourth is not, because its member \(\set{\set{\emptyset}}\) is non-transitive.

The elements of an ordinal are not only transitive sets: they are in fact
ordinals again, as shown by the following standard argument.

\begin{lemma}[\flink{Lemma-19}]%
  \label{being-set-theoretic-ordinal-is-hereditary}
  Being an ordinal is hereditary: the elements of a set-theoretic ordinal are
  themselves ordinals.
\end{lemma}
\begin{proof}
  Let \(x\) be a set-theoretic ordinal and \(y \in x\). Then \(y\) is a
  transitive set by assumption. Moreover, if \(z \in y\), then \(z\) is again a
  transitive set, because \(z \in x\) by transitivity of \(x\).
\end{proof}

\subsection{Set theory in homotopy type theory}
In order to relate the set-theoretic and type-theoretic approaches to ordinals,
we recall a higher inductive~\cite[\S{}6]{HoTTBook} construction of a model
of constructive set theory inside homotopy type theory
from~\cite[\S{}10.5]{HoTTBook}.
The model may be seen as a refinement of Aczel's~\cite{Aczel1978} interpretation
of constructive set theory in type theory, and is referred to as the
\emph{cumulative hierarchy} in the HoTT book~\cite{HoTTBook} and the
\emph{iterative hierarchy} in Gylterud~\cite{Gylterud2018}.

It is convenient to introduce the following terminology before proceeding.

\begin{definition}[\flink{Definition-20} Equal images]
  Two maps \(f : A \to X\) and \(g : B \to X\) with the same codomain are said
  to have \emph{equal images} if for every \(a : A\), there exists some
  \(b : B\) such that \(f \, a = g \, b\), and conversely, for every \(b : B\),
  there exists some \(a : A\) with \(g \, b = f \, a\).
\end{definition}

\begin{definition}[\flink{Definition-21} Cumulative hierarchy \(\V\); {\cite[Def~10.5.1]{HoTTBook}}]
  The \emph{cumulative hierarchy} \(\V\) with respect to a type universe \(\U\)
  is the higher inductive type with the following constructors:
  \begin{enumerate}[(i)]
  \item for every type \(A : \U\) and \(f : A \to \V\) we have an element of
    \(\V\), denoted by \(\Vset{A}{f}\);
  \item for every two types \(A, B : \U\) and maps \(f : A \to \V\) and
    \(g : B \to \V\), if \(f\) and \(g\) have equal images, then we have an
    identification \(\Vset{A}{f} = \Vset{B}{g}\);
  \item set-truncation, i.e., for every \(x,y : \V\) and \(p,q : x = y\), we
    have an identification \(p = q\).
  \end{enumerate}
\end{definition}

\(\V\) is a model of set theory by \cite[Thm~10.5.8]{HoTTBook}.
It is instructive to see how to represent the sets \(\emptyset\),
\(\{\emptyset\}\) and \(\{\emptyset,\{\emptyset\}\}\)
from~\cref{ex:transitive-sets} in \(\V\):
\begin{itemize}
\item The empty set \(\emptyset\) is represented as
  \(\Quinequote{\emptyset} \colonequiv \Vset{\Zero}{!}\) where \(!\) is the
  unique map from \(\Zero\) to \(\V\).
\item The singleton set \(\{\emptyset\}\) may be represented by setting
  \(\Quinequote{\{\emptyset\}} \colonequiv
  \Vset{\One}{\lambdadot{\star}{\Quinequote{\emptyset}}}\).
\item Finally, the set \(\{\emptyset,\{\emptyset\}\}\) can be encoded as
  \(\Vset{\Two}{f}\) where \(f(0) \colonequiv \Quinequote{\emptyset}\) and
  \(f(1) \colonequiv \Quinequote{\{\emptyset\}}\).
\end{itemize}

The second constructor of \(\V\) ensures that the elements have the correct
notion of equality. For instance, using the example given directly above, it
means that the elements \(\Vset{\Two}{f}\) and
\(\Vset{\mathbb N}{f \circ \mathsf{isEven}}\) are equal.

Observe that \(\V\) is a large type, i.e., it lives in the next universe
\(\U^+\). Following~\cite[\S{}10.5]{HoTTBook}, we now define the set
membership and the subset relation on \(\V\), so that we can define
set-theoretic ordinals inside \(\V\).

\begin{definition}[\flink{Definition-22} Set membership \({\in}\) on \(\V\)]
  We define the \emph{set membership} relation
  \({\in} : \V \to \V \to \Prop_{\U^+}\) inductively as:
  \[
    x \in \Vset{A}{f} \;\colonequiv\; \exists(a : A). {f \, a = x}.
  \]
  This is well-defined because \(\Prop_{\U^+}\) is a set (in the sense of HoTT),
  and if \(f\)~and~\(g\) have equal images, then \({x \in \Vset{A}{f}}\) holds
  exactly when \(x \in \Vset{B}{g}\) does.
\end{definition}

\begin{definition}[\flink{Definition-23} Subset relation \(\subseteq\)]
  We define the \emph{subset relation}
  \({\subseteq} : {\V \to \V \to \Prop_{\U^+}}\) as
  \[
    x \subseteq y \;\colonequiv\; \forall (v : \V).{v \in x \to v \in y}.
  \]
\end{definition}

The type \(\V\) models Myhill's Constructive Set Theory~\cite{Gylterud2018}, and
in fact all of Zermelo-Fraenkel set theory with Choice, if we assume the axiom
of choice in type theory~\cite{Eleftheriadis2021}. In the following, we will in
particular need the following two set-theoretic axioms:

\begin{lemma}[\flink{Lemma-24} Items~(i) and (vii) of {\cite[Thm~10.5.8]{HoTTBook}}]%
  \label{extensionality-and-induction}
  The following two set-theoretic axioms are satisfied by \(\V\):
  \begin{enumerate}[(i)]
  \item \emph{extensionality:} two elements \(x\) and \(y\) of \(\V\) are equal
    if and only if \(x \subseteq y\) and \(y \subseteq x\), and
  \item \emph{\(\in\)-induction:} for any prop-valued family
    \(P : \V \to \Prop\), if, for every \(x : \V\), we have \(P(x)\) whenever
    \(P(y)\) holds for all \(y \in x\), then \(P\) holds at every element of
    \(\V\). \qedNoProof
  \end{enumerate}
\end{lemma}

The set membership relation allows us to formulate the set-theoretic notions
of~\cref{sec:ordinals-in-set-theory} for~\(\V\), and hence, to define the type
of set-theoretic ordinals in \(\V\).

\begin{definition}[\flink{Definition-25} Type of set-theoretic ordinals]
  The \emph{type \(\Vord\) of set-theoretic ordinals} is the \(\Sigma\)-type of
  those \(x : \V\) such that \(x\) is a set-theoretic ordinal in the sense
  of~\cref{set-theoretic-ordinal}.
\end{definition}

The subtype of set-theoretic ordinals is then an example of a type-theoretic
ordinal, which we show to be equal to the type \(\Ord\) of type-theoretic
ordinals in the next subsection.

\begin{theorem}[\flink{Theorem-26}]
  Set membership makes \(\Vord\) into a type-theoretic ordinal.
\end{theorem}
\begin{proof}
  Wellfoundedness follows from \(\in\)-induction, and set membership is a
  transitive relation on \(\Vord\): if we have set-theoretic ordinals
  \(x,y,z : \Vord\) such that \(x \in y\) and \(y \in z\), then \(x \in z\),
  because \(z\) is a transitive set.
  For extensionality, assume that we have \(x,y : \Vord\) such that
  \(u \in x \iff u \in y\) for every 
  \(u : \Vord\). We need to show that \(x = y\). By extensionality
  in the sense of~\cref{extensionality-and-induction}, it suffices to show that
  \(v \in x \iff v \in y\) for all \(v : \V\).
  But if \(v \in x\), then \(v : \Vord\), because being a set-theoretic ordinal
  is hereditary (\cref{being-set-theoretic-ordinal-is-hereditary}). Hence,
  \(v \in y\) by assumption. Similarly, \(v \in y\) implies \(v \in x\), so that
  \(x = y\), as desired.
\end{proof}

\subsection{Set-theoretic and type-theoretic ordinals coincide}
Having reviewed the necessary preliminaries, we prove in this subsection that
the set-theoretic and type-theoretic ordinals coincide. More precisely, we
construct an isomorphism of type-theoretic ordinals between \(\Vord\) and
\(\Ord\) by constructing maps in both directions.

\begin{definition}[\flink{Definition-27} \(\Phi\)]\label{def:Phi}
  The map \(\Phi : \Ord \to \V\) is defined by transfinite recursion on
  \(\Ord\) as
  \[
    \Phi(\alpha) \;\colonequiv\; \Vset{\alpha}{\lambdadot{a}{\Phi(\alpha
          \initseg a)}}
  \]
\end{definition}

The function \(\Phi\) is well-defined, because for every \(a : \alpha\), the initial segment
\(\alpha \initseg a\) is strictly smaller than \(\alpha\), as
ordinals.

\begin{lemma}[\flink{Lemma-28}]\label{Phi-properties}
  The map \(\Phi\) is injective and preserves and reflects the strict and weak
  orders, i.e., for every two type-theoretic ordinals \(\alpha,\beta : \Ord\),
  we have
  \begin{enumerate}[(i)]
  \item \(\alpha = \beta \;\iff\;\Phi \, \alpha = \Phi \, \beta\),
  \item \(\alpha < \beta \;\iff\; \Phi \, \alpha \in \Phi \, \beta\), and
  \item \(\alpha \leq \beta \;\iff\; \Phi \, \alpha \subseteq \Phi \, \beta\).
  \end{enumerate}
\end{lemma}
\begin{proof}
  That $\Phi$ preserves equality is automatic.
  \paragraph{\(\alpha < \beta \Rightarrow \Phi \, \alpha \in \Phi \, \beta\)}
  If \(\alpha < \beta\), then we have \(b : \beta\) such that
  \(\alpha = \beta \initseg b\). Hence, in this case, we have
  \(\Phi \, \alpha = \Phi(\beta \initseg b)\), viz.\
  \(\Phi \, \alpha \in \Phi \, \beta\) by the definitions of \({\in}\)
  and~\(\Phi\).

  \paragraph{\(\alpha \leq \beta \Rightarrow \Phi \, \alpha \subseteq \Phi \, \beta\)}
  For \(x \in \Phi \, \alpha\), we get an $a$ with \(x = \Phi(\alpha \initseg a)\)
  by definition. \cref{simulation-properties} gives us $b$ such that
  \(\alpha \initseg a = \beta \initseg b\), and hence \(x \in \Phi \, \beta\),
  as desired.

  \paragraph{Injectivity -- \(\Phi \, \alpha = \Phi \, \beta \Rightarrow \alpha = \beta\)}
  We do transfinite induction on
  \(\Ord\). Assume \(\alpha : \Ord\) and the induction hypothesis:
  for every element \(a : \alpha\) and ordinal \(\beta : \Ord\), if
  \(\Phi(\alpha \initseg a) = \Phi \, \beta\), then
  \(\alpha \initseg a = \beta\).  We must prove that
  \(\Phi \, \alpha = \Phi \, \beta\) implies \(\alpha = \beta\) for all ordinals
  \(\beta : \Ord\).
  So assume that \(\beta : \Ord\) is such that
  \(\Phi \, \alpha = \Phi \, \beta\). We show that \(\alpha \leq \beta\); the
  reverse inequality is proved similarly. By~\cref{simulation-properties}, it
  suffices to prove that \(\alpha \initseg a < \beta\) for every \(a : \alpha\).
  For such \(a : \alpha\) we have
  \(\Phi(\alpha \initseg a) \in \Phi \, \alpha = \Phi \, \beta\), and hence, there
  exists some \(b : \beta\) with
  \(\Phi(\alpha \initseg a) = \Phi(\beta \initseg b)\). Our induction hypothesis
  then yields \(\alpha \initseg a = \beta \initseg b\), and hence the desired
  \(\alpha \initseg a < \beta\).

  \paragraph{\(\Phi \, \alpha \in \Phi \, \beta \Rightarrow \alpha < \beta\)}
  If \(\Phi \, \alpha \in \Phi \, \beta\), then there exists some \(b : \beta\)
  with \(\Phi \, \alpha = \Phi(\beta \initseg b)\), and hence \(\alpha < \beta\)
  by injectivity of \(\Phi\).

  \paragraph{\(\Phi \, \alpha \subseteq \Phi \, \beta \Rightarrow \alpha \leq \beta\)}
  Suppose \(\Phi \, \alpha \subseteq \Phi \, \beta\). Then for every
  \(a : \alpha\), there exists some \(b : \beta\) with
  \({\Phi(\alpha \initseg a) = \Phi(\beta \initseg b)}\).
  Injectivity of $\Phi$ and \cref{simulation-properties}
  imply \(\alpha \leq \beta\).
\end{proof}

\begin{lemma}[\flink{Lemma-29}]\label{Phi-factors-through-Vord}
  The map \(\Phi : \Ord \to \V\) factors through the inclusion
  \(\Vord \hookrightarrow \V\).
\end{lemma}
\begin{proof}
  We first show directly that \(\Phi \, \alpha\) is a transitive set for every
  \(\alpha : \Ord\): if we have \(x,y : \V\) with
  \(x \in y \in \Phi \, \alpha\), then there exists \(a : \alpha\) with
  \(x = \Phi(\alpha \initseg a)\) and hence \(b : \alpha \initseg a\) with
  \(y = \Phi\pa*{\pa*{\alpha \initseg a} \initseg b}\). But
  \(\pa{\alpha \initseg a} \initseg b\) and \(\alpha \initseg b\) are equal
  ordinals by~\cref{iterated-initial-segments}, so
  \(y = \Phi(\alpha \initseg b)\) and thus \(y \in \Phi \, \alpha\), as desired.

  Now we prove that \(\Phi \, \alpha\) is a set-theoretic ordinal for every
  \(\alpha : \Ord\) by transfinite induction on \(\Ord\). We just established
  that \(\Phi \, \alpha\) is a transitive set and if \(x \in \Phi \, \alpha\), then
  \(x = \Phi(\alpha \initseg a)\) for some \(a : \alpha\), so that \(x\) must be
  a transitive set by the induction hypothesis.
\end{proof}

Thus, one half of the desired isomorphism is given by
\({\Phi : \Ord \to \Vord}\). We define a map in the other direction now.

\begin{definition}[\flink{Definition-30} \(\Psi\)]\label{def:Psi}
  We define \(\Psi : \V \to \Ord\) recursively by
  \[
    \Psi(\Vset{A}{f}) \;\colonequiv\; \bigvee_{a : A}\pa*{\Psi(f \, a) + \One}.
  \]
  This map is well-defined because \(\Ord\) is a set, and if \(f\)~and~\(g\)
  have equal images then the suprema \(\Psi(\Vset{A}{f})\) and
  \(\Psi(\Vset{B}{g})\) are seen to coincide.
\end{definition}

\begin{remark}
  This function above assigns the \emph{rank} to a set and is well-known in set
  theory, see for example~\cite[p.~743]{Aczel1977}
  and~\cite[Def~9.3.4]{AczelRathjen2010}.
\end{remark}

\begin{proposition}[\flink{Proposition-32}]\label{Psi-section-of-Phi}
  When restricted to \(\Vord\), the map \(\Psi\) is a section of \(\Phi\), i.e.,
  for 
  \(x : \Vord\), we have \(\Phi(\Psi \, x) = x\).
\end{proposition}
\begin{proof}
  Since we are proving a proposition, the induction principle of \(\V\) implies
  that it suffices to prove that for every \(A : \U\) and \(f : A \to \V\) such
  that \(\Vset{A}{f}\) is a set-theoretic ordinal, the equality
  \(\Phi(\Psi(\Vset{A}{f})) = \Vset{A}{f}\) holds, assuming the \emph{induction
    hypothesis}: \(\Phi(\Psi(f \, a)) = f \, a\) holds for all \(a : A\).
  (Note that every \(f \, a\) is a set-theoretic ordinal if \(\Vset{A}{f}\) is.)
  We compute that
  \[
    \Phi(\Psi(\Vset{A}{f})) = \Vset{s}{\lambdadot{y}{\Phi\pa*{s \initseg y}}},
  \]
  where \(s \colonequiv \bigvee_{a : A}\pa*{\Psi(f \, a) + \One}\).
  We now use the second constructor of \(\V\) to prove that
  \(\Vset{s}{\lambdadot{y}{\Phi\pa*{s \initseg y}}}\) is equal to
  \(\Vset{A}{f}\), i.e., we show that \(\lambdadot{y}{\Phi\pa*{s \initseg y}}\)
  and \(f\) have the same image.
  It is convenient to set up some notation: we write \(c_a\) for
  \(\Psi(f \, a) + \One\).

  In one direction, suppose that \(a : A\), then
  \[
    f \, a = \Phi(\Psi(f \, a)) = \Phi\pa*{{c_a} \initseg {\inr \star}} =
    \Phi\pa*{s \initseg [a,\inr \star]},
  \]
  where the first equality holds by induction hypothesis and the second and
  third by~\cref{initial-segments-of-sums,initial-segments-of-sups},
  respectively.

  Conversely, if we have \(y : s\), then by~\cref{initial-segments-of-sups}
  there exist some \(a : A\) and \(w : c_a\) such that
  \(s \initseg y = c_a \initseg w\).
  There are now two cases: either \(w = \inr \star\) or \(w = \inl x\) with
  \(x : \Psi(f \, a)\).
  If \(w = \inr \star\), then, as before,
  \[
    \Phi(s \initseg y) = \Phi(c_a \initseg \inr \star) = \Phi(\Psi(f \, a))
    = f \, a.
  \]

  So suppose that \(w = \inl x\) with \(x : \Psi(f \, a)\). It is here that we
  use our assumption that \(\Vset{A}{f}\) is a set-theoretic ordinal.
  Indeed, since \(\Psi(f \, a) \initseg x\) is an initial segment of
  \(\Psi(f \, a)\), we have
  \(\Phi(\Psi(f \, a) \initseg x) \in \Phi(\Psi(f \, a)) = f \, a\)
  by~\cref{Phi-properties} and the induction hypothesis.
  But \(f \, a \in \Vset{A}{f}\) and the latter is a transitive set, so
  \(\Phi(\Psi(f \, a) \initseg x) \in \Vset{A}{f}\).
  By definition of set membership, this means that there exists some \(a' : A\)
  with \(\Phi(\Psi(f \, a) \initseg x) = f(a')\).
  Finally,
  \[
    \Phi(s \initseg y) = \Phi(c_a \initseg \inl x) = \Phi(\Psi(f \, a) \initseg x) = f(a'),
  \]
  where the second equality holds by~\cref{initial-segments-of-sums}.
  Hence, 
  \(f\) and \(\Phi(s \initseg -)\) have the same image,
  completing the proof.
\end{proof}

We are now ready to prove the main theorem of~\cref{sec:ordinals}: the
type-theoretic and set-theoretic ordinals coincide.
\begin{theorem}[\flink{Theorem-33}]\label{Ord-and-Vord-coincide}
  The ordinals \(\Ord\) and \(\Vord\) are isomorphic (as type-theoretic
  ordinals). Hence, by univalence, they are equal.
\end{theorem}
\begin{proof}
  By~\cref{Phi-factors-through-Vord} we have a map \(\Phi : \Ord \to
  \Vord\). Moreover, it is an injection by~\cref{Phi-properties} and a (split)
  surjection by~\cref{Psi-section-of-Phi}. Hence, \(\Phi\) is a
  bijection. But~\cref{Phi-properties} tells us that it also preserves and
  reflects the strict orders, so it is an isomorphism of ordinals
  by~\cref{iso-of-ordinals}.
\end{proof}

\subsection{Revisiting the rank of a set}\label{subsec:revisit-rank}

The recursive nature of the map \(\Psi\) from~\cref{def:Psi} that computes the rank
of a set in \(\V\) is convenient for proving properties by induction. It is
possible, however, to give a conceptually simpler and non-recursive
description, although this requires paying close attention to size issues.

\begin{definition}[\flink{Definition-34} Type of elements, \(\totalspace x\)]
  Given an element \({x : \V}\), we write \(\totalspace x\) for its type of
  elements, i.e.,
  \[
    \totalspace x \;\colonequiv\; \Sigma(y : \V). {y \in x}.
  \]
\end{definition}

\begin{proposition}[\flink{Proposition-35}]
  If \(x : \V\) is a set-theoretic ordinal, then \(\totalspace x\) ordered by
  \({\in}\) is a type-theoretic ordinal.
\end{proposition}
\begin{proof}
  Since being a set-theoretic ordinal is hereditary, we have
  \(\totalspace x = \Sigma(y : \Vord). {y \in x}\), so that the
  former inherits the ordinal structure from \(\Vord\).
\end{proof}

It now becomes important to pay close attention to type universe parameters, so
we will annotate them with subscripts. Notice that the \(\totalspace\)-operation
does not define a map \(\V_{\U} \to \Ord_{\U}\) like \(\Psi\) does, but rather a
map \(\V_{\U} \to \Ord_{\U^+}\), because the cumulative hierarchy \(\V_\U\) with
respect to the universe \(\U\) is itself a type in the next universe \(\U^+\).

Still, we will prove that \(\Psi(x)\) and \(\totalspace x\) are
\emph{isomorphic} ordinals for every set-theoretic ordinal \(x : \V_\U\), even
though they cannot be equal due to their different sizes.
However, we can do a bit better by observing, as in~the HoTT book~\cite[Lem~10.5.5]{HoTTBook},
that the cumulative hierarchy is locally small (in the sense of
Rijke~\cite{Rijke2017}), meaning its identity types are \(\U\)-valued up to
equivalence. Then we observe that \(\totalspace(\Vset{A}{f})\) is equal to the
image of \(f\), which is equivalent to a type in \(\U\) thanks to the fact that
\(\V\) is a locally small set.
This general fact on small images of maps into locally small sets is a ``set
replacement principle'', discussed by Rijke~\cite{Rijke2017} and
de Jong and Escard\'o~\cite{deJongEscardo2022}.
Specifically, the image of \(f\) is equivalent to the set quotient \(A/{\sim}\),
where \(A\) is the domain of \(f\) and \({\sim}\) relates two elements if \(f\)
identifies them. We then make the quotient \(A/{\sim}\) into an ordinal by
defining \([a] < [b]\) as \(f \, a \in f \, b\).
Finally, we can resize \(A/{\sim}\) to an ordinal in \(\U\) by using that
\(\V\) is locally small and by employing a \(\U\)-valued membership relation,
as explained below.

We stress that none of the above constructions rely on propositional resizing
principles.

\subsubsection{The cumulative hierarchy is locally small}
We again follow~\cite[\S{}10.5]{HoTTBook} in defining a recursive
bisimulation relation that makes \(\V\) a locally small type.

\begin{definition}[\flink{Definition-36} Bisimulation {\cite[Def~10.5.4]{HoTTBook}}]
  The \emph{bi\-simulation} relation
  \({\approx} : \V_\U \to \V_\U \to \Prop_\U\) is inductively defined by
  \begin{align*}
    {\Vset{A}{f}} \approx {\Vset{B}{g}} &\colonequiv
    \pa*{\forall (a : A).\exists (b : B).{f \, a \approx g \, b}} \\
    &\hspace{3pt}\times
    \pa*{\forall (b : B).\exists (a : A).{g \, b \approx f \, a}}.
  \end{align*}
\end{definition}

\begin{lemma}[\flink{Lemma-37} Lemma~10.5.5 of \cite{HoTTBook}]%
  \label{cumulative-hierarchy-is-locally-small}
  For every \(x,y : \V_\U\), we have an equivalence of propositions
  \(\pa*{x = y} \simeq \pa*{x \approx y}\). \qedNoProof
\end{lemma}
Hence, the bisimulation relation captures equality on \(\V\), but has the
advantage that it has values in \(\U\) rather than \(\U^+\).
This also allows us to define a \(\U\)-valued membership relation.

\begin{definition}[\flink{Definition-38} \({\insmall}\)]
  Define \({\insmall} : \V_\U \to \V_\U \to \Prop_\U\) inductively by
  \(
    x \insmall \Vset{A}{f} \colonequiv \exists (a : A).{f \, a \approx x}.
  \)
\end{definition}

\begin{lemma}[\flink{Lemma-39}]
  For every \(x,y : \V_\U\), we have an equivalence of propositions
  \(\pa*{x \in y} \simeq \pa*{x \insmall y}\).
\end{lemma}
\begin{proof}
  By \(\V\)-induction and \cref{cumulative-hierarchy-is-locally-small}.
\end{proof}

\subsubsection{The set quotients}

Throughout this subsection, assume that we are given \(A : \U\) and
\(f : A \to \V\) such that \(\Vset{A}{f}\) is a set-theoretic ordinal.
We show that the type of elements of \(\Vset{A}{f}\) is given by a suitable
quotient of~\(A\). This simple quotient can capture all the elements of
\(\Vset{A}{f}\) precisely because \(\Vset{A}{f}\) is hereditarily transitive.

\begin{definition}[\flink{Definition-40}]
  We write \(A/{\sim}\) for the set quotient of \(A\) by the \(\U^+\)-valued
  equivalence relation
  \(
    a \sim b \colonequiv (f \, a = f \, b).
  \)
  Similarly, we write \(A/{\sim_\U}\) for the set quotient of \(A\) by the
  \(\U\)-valued equivalence relation given by
  \(
    a \sim_{\U} b \colonequiv (f \, a \approx f \, b).
  \)
\end{definition}

The important thing to note in the above definition is that \(A/{\sim} : \U^+\),
while \(A/{\sim_{\U}} : \U\).  It is easy to prove that the latter is a
\emph{small replacement} of the former:

\begin{lemma}[\flink{Lemma-41}]\label{quotient-is-image-qua-sets}
  Writing \(\im f \colonequiv \Sigma (v : \V).\exists (a : A).{f \, a = v}\) for
  the image of \(f\), we have \((A/{\sim_\U}) \simeq (A/{\sim}) = (\im f)\). \qedNoProof
\end{lemma}

We define relations on the quotients that make them into large and
small type-theoretic ordinals, respectively.

\begin{definition}[\flink{Definition-42}]
  We define a \(\U^+\)-valued binary relation \(\prec\) on \(A/{\sim}\) by
  \(
    [a] \prec [b] \colonequiv (f \, a \in f \, b).
  \)
  Similarly, we define a \(\U\)-valued relation \(\prec_{\U}\) on
  \(A/{\sim_{\U}}\) by
  \(
    [a] \prec_{\U} [b] \colonequiv (f \, a \insmall f \, b).
  \)
\end{definition}

\begin{proposition}[\flink{Proposition-43}]
  The relation \(\prec\) makes \(A/{\sim}\) into an ordinal in~\(\U^+\), and
  \(\prec_{\U}\) makes \(A/{\sim_\U}\) into an ordinal in~\(\U\).
\end{proposition}
\begin{proof}
  For transitivity, it suffices to prove that \([a] \prec [b]\) and
  \([b] \prec [c]\) together imply \([a] \prec [c]\) for all \(a,b,c : A\). But
  this follows from the fact that \(f(c)\) is a transitive set which holds
  because it is an element of the set-theoretic ordinal \(\Vset{A}{f}\).
  For extensionality, assume that \({x \prec [a] \iff x \prec [b]}\) for every
  \(x : A/{\sim}\). We have to prove that \([a] = [b]\), i.e., that
  \(f \, a = f \, b\). We show that \(f \, a \subseteq f \, b\) and note that
  the reverse inclusion is proved similarly.
  Suppose that we have \(x : \V\) with \(x \in f \, a\). Then because \(f \, a\) is
  a member of the transitive set \(\Vset{A}{f}\), we get \(x \in
  \Vset{A}{f}\). Hence, there exists some \(c : A\) with \(f(c) = x\). But then
  \(f(c) = x \in f \, a\), and so \([c] \prec [a]\). Hence, \([c] \prec [b]\) by
  assumption, and therefore, \(x = f(c) \in f \, b\), as desired.
  Further, to see that every element of \(A/{\sim}\) is accessible, we prove the
  following statement by transfinite induction in the ordinal \((\V,{\in})\):
  for every \(x : \V\) and every \(a : A\), if \(f \, a = x\), then \([a]\) is
  accessible.
  So let \(x : \V\) and \(a : A\) be such that \(f \, a = x\) and assume the
  induction hypothesis that for every \(y \in x\) and \(b : A\), if
  \(f \, b = y\), then \([b]\) is accessible.
  For accessibility of \([a]\), it suffices to prove that every \([b]\) is
  accessible whenever we have \(b : A\) with \([b] \prec [a]\).
  But given such a \(b : A\) we have \(f \, b \in f \, a = x\), and hence
  accessibility of \([b]\) by induction hypothesis.
  The claim about \(\pa*{A/{\sim_\U},{\prec_\U}}\) is proved
  analogously.
\end{proof}

Finally, the quotient is equal to the type of elements:

\begin{lemma}[\flink{Lemma-44}]\label{quotient-is-image-qua-ordinals}
  For every \(A : \U\) and \(f : A \to \V\), the ordinals \(\pa*{A/{\sim},{\prec}}\) and
  \(\pa*{\totalspace(\Vset{A}{f}),{\in}}\) are equal.
\end{lemma}
\begin{proof}
  By~\cref{iso-of-ordinals}, we only need to verify that the bijection from
  \cref{quotient-is-image-qua-sets} preserves and reflects the order,
  but this is clear because \([a] \prec [b]\) holds exactly when
  \(f \, a \in f \, b\).
\end{proof}

\subsubsection{Alternative descriptions of the rank}

We are now ready to prove the main result of this subsection: we show that the
rank of \(\Vset{A}{f}\), as recursively computed by \(\Psi\), is equal to the
quotient \(A/{\sim_\U}\), thus providing a simpler non-recursive description of
its rank.

\begin{theorem}[\flink{Theorem-45}]\label{rank-as-quotient}
  The ordinals \(\Psi(\Vset{A}{f})\) and \(A/{\sim_\U}\) are equal.
\end{theorem}
\begin{proof}
  Because \(\Phi\) is injective with inverse \(\Psi\)
  (\cref{Phi-properties,Psi-section-of-Phi}), it suffices to show that
  \[
    \Phi\pa*{A/{\sim_\U}} = \Vset{A}{f}.
  \]
  By definition of \(\Phi\) and equality on \(\V\), it is enough to prove
  \[
    \Phi\pa*{A/{\sim_\U} \initseg [a]} = f \, a
  \]
  for every \(a : A\).
  We slightly generalize this statement so that it becomes amenable to a proof
  by transfinite induction on \(A/{\sim_\U}\). Namely, we show that for every
  \(a' : A/{\sim_\U}\) and every \(a : A\), if \(a' = [a]\), then
  \(\Phi\pa*{A/{\sim_\U} \initseg [a]} = f \, a\) holds.
  So suppose that we have \(a : A\). We first show that
  \(f \, a \subseteq \Phi\pa*{A/{\sim_\U} \initseg [a]}\).
  Now if \(x \in f \, a\), then there exists some \(b : A\) with \(x = f \, b\),
  because \(f \, a\) is a member of the transitive set \(\Vset{A}{f}\).
  But then \(f \, b = x \in f \, a\), so \([b] \prec [a]\) and hence
  \(\Phi\pa*{A/{\sim_\U} \initseg [b]} = f \, b = x\) by the induction hypothesis.
  Further, \(\Phi\pa*{A/{\sim_\U} \initseg [b]}\) is an element of
  \(\Phi\pa*{A/{\sim_\U} \initseg [a]}\), because \([b] \prec [a]\) and
  \(\pa*{A/{\sim_\U} \initseg [a]} \initseg [b] = A/{\sim_\U} \initseg [b]\)
  by~\cref{iterated-initial-segments}.
  Hence,
  \(x = \Phi\pa*{A/{\sim_\U} \initseg [b]} \in \Phi\pa*{A/{\sim_\U} \initseg
    [a]}\), as desired.
  For the other inclusion, suppose that
  \(x \in \Phi\pa*{A/{\sim_\U} \initseg [a]}\).
  By another application of~\cref{iterated-initial-segments}, we see that there
  exists some \(b : A\) such that \([b] \prec [a]\) and
  \(x = \Phi\pa*{A/{\sim_\U} \initseg [b]}\).
  Then \(x = \Phi\pa*{A/{\sim_\U} \initseg [b]} = f \, b\) by the induction hypothesis,
  but also \([b] \prec [a]\), so that \(x = f \, b \in f \, a\), as we wished to show.
\end{proof}

\begin{corollary}[\flink{Corollary-46}]\label{cor:rank-nonrecursively}
  For every \(x : \V_\U\), the ordinals \(\Psi(x)\) and \(\totalspace x\) are
  isomorphic, but not equal, because the latter lives in a larger universe.
\end{corollary}
\begin{proof}
  Since we are proving a proposition, \(\V\)-induction implies that it is enough
  to prove that \(\Psi(\Vset{A}{f})\) and \(\totalspace(\Vset{A}{f})\) are
  isomorphic ordinals, for every \(A : \U\) and \(f : A \to \V\).
  But this holds by the following chain of isomorphisms of ordinals:
  \begin{align*}
    \Psi(\Vset{A}{f})
    &\,\simeq\, A/{\sim_\U} &&\text{(by~\cref{rank-as-quotient})} \\
    &\,\simeq\, A/{\sim}  &&\text{(by~\cref{quotient-is-image-qua-sets})} \\
    &\,\simeq\, \totalspace(\Vset{A}{f}) &&\text{(by~\cref{quotient-is-image-qua-ordinals})}.
    \qedhere
  \end{align*}
\end{proof}

\section{Generalizing from ordinals to sets}\label{sec:mewo}

Since we now understand the subtype of $\V$ that consists of exactly the
hereditarily transitive sets, it is a natural goal to characterize the full type
$\V$
\begin{wrapfigure}{l}{0.16\textwidth}
\begin{equation*} 
  \begin{tikzcd}
    \Vord \ar[d,hookrightarrow] \ar[r,"\simeq"] & \Ord \ar[d,hookrightarrow] \\
    \V \ar[r,"\simeq"] & ?
  \end{tikzcd}
\end{equation*}
\end{wrapfigure}
to complete the square on the left, by generalizing the notion of type-theoretic
ordinals.
Since arbitrary $\V$\nobreakdash-sets are not necessarily transitive, we
certainly need to give up transitivity.  However, as discussed in the
introduction, doing so and simply considering \emph{extensional wellfounded
  relations} is insufficient to complete the square.  Our solution is to further
equip them with \emph{covering markings}.  We then develop a generalization of
the theory of type-theoretic ordinals that matches $\V$.

Giving up transitivity as an assumption, we at times need to consider the
transitive and reflexive-transitive closure of a given relation $<$, i.e., the
smallest proposition-valued such relations that include $<$.  We denote them by
$\tc$ and $\trc$ respectively.  In type theory, it is standard to implement $\tc$ and $\trc$ using inductive
families describing sequences of steps, which then can be propositionally
truncated to ensure proof-irrelevance.

\subsection{Mewos: marked extensional wellfounded order relations}

We start by defining the generalization of type-theoretic ordinals that we need
to complete the above square.

\begin{definition}[\flink{Definition-47} Mewo]\label{def:ewos}
  A \emph{marked extensional wellfounded order} (\emph{mewo}) is a triple
  $(X,<,\mar)$, where $X$ is a type, $<$ is a binary relation on $X$ that is
  extensional, wellfounded, and valued in propositions, i.e., $(X, <)$ is an
  ensemble in the sense of Taylor~\cite{Taylor1996}, and $\mar : X \to \Prop$ is
  a prop-valued predicate on $X$ (called a \emph{marking}).

  We say that $x:X$ is \emph{marked} if $\mar(x)$, and \emph{covered} if there
  exists a marked $x_0$ such that $x \trc x_0$.
  A \emph{covered mewo} is a mewo where every element is covered.
\end{definition}

We write $\mewo$ for the type of mewos, and $\cewo$ for its subtype of covered
mewos.  From now on, we keep the order and the marking implicit, overloading the
symbols $<$ and $\mar$ whenever required, and denote a mewo only by its carrier
$X$.  The \emph{subtype of marked elements} of $X$ is the total space of $\mar$,
\[
  \Mar_X \defeq \Sigma(x:X). \mar(x)
\]
and we implicitly apply the first projection to treat elements of $\Mar_X$ as
elements of $X$.  With this convention, a mewo is covered if we can show
\[
  \forall(x:X). \exists (x_0:\Mar_X). x \trc x_0.
\]

\begin{remark}\label{rem:mewo-equality-small}
  As for ordinals, the extensionality of the relation $<$ implies that $X$ is necessarily a set.
  Further note that, by univalence, an equality $X = Y$ between mewos is an
  equivalence $e : X \simeq Y$ that preserves and reflects both order and
  marking, i.e., satisfies $(x_1 < x_2) \leftrightarrow (e \, x_1 < e \, x_2)$
  and $\mar(x) \leftrightarrow \mar(e \, x)$.  The identical characterization
  holds if the mewos in consideration are covered since coveredness is a
  propositional property.
\end{remark}

Our second main result is that $\cewo$ is the missing corner in the discussed
square as shown on the right, cf.
\begin{wrapfigure}[7]{r}{0.25\textwidth}
  \begin{equation}\label{eq:square-completed}
    \begin{tikzcd}
      \Vord \ar[d,hookrightarrow] \ar[r,"\simeq"] & \Ord \ar[d,hookrightarrow] \\
      \V \ar[r,"\simeq"] & \cewo
    \end{tikzcd}
  \end{equation}
\end{wrapfigure}
\cref{thm:second-main-result}.  This means that a covered mewo simultaneously
behaves like a generalized type-theoretic ordinal \emph{and} a set in $\V$.  The
first connection is easy to make precise:

\begin{example}[Ordinals as covered mewos]\label{ex:canonical-ewo-from-ord}
  Given a type-theoretic ordinal, we get a covered mewo by forgetting
  the transitivity of the order and marking everything.
\end{example}

We will later see that, if we view a mewo as a $\V$-set, it is exactly the
marked elements that become elements of the set (while the unmarked ones become
elements of elements of \ldots).  Therefore, for type-theoretic ordinals,
everything will be an element of the corresponding $\V$-set.  This is already
determined by the top horizontal map in \eqref{eq:square-completed}, i.e., the
map $\Phi$ from \cref{def:Phi}.  Based on this observation,
\cref{ex:canonical-ewo-from-ord} guides and motivates much of our theory of
mewos.

\subsection{Order relations between mewos}

The main concepts that we need to generalize from type-theoretic ordinals are
the relations $\leq$ and $<$ between mewos. The above square
\eqref{eq:square-completed} means that these relations necessarily need to
correspond to the relations $\subseteq$ and $\in$ between $\V$-sets.  To begin,
the concept of a \emph{simulation} between type-theoretic ordinals is
straightforward to generalize to mewos:
\begin{definition}[\flink{Definition-50} Simulation, \({\leq}\)]
  Given mewos $X$ and $Y$, a function $f : X \to Y$ is a \emph{simulation} if it
  fulfills the following properties:
  \begin{enumerate}[(i)]
  \item it preserves the markings: $\forall x. \mar(x) \to \mar(f\, x)$;
  \item it is monotone: $x_1 < x_2 \to f \, x_1 < f \, x_2$;
  \item it has the initial segment property, i.e., its image is downwards closed
    in a strong sense:
    \begin{equation}\label{eq:simu-property-for-mewo}
      \forall x_2. \forall (y < f\, x_2). \exists(x_1 < x_2). f \, x_1 = y.
    \end{equation}
  \end{enumerate}
  We write $X \leq Y$ for the type of simulations.
\end{definition}

An example of a function that fails to be a simulation precisely because it does
not preserve markings is the identity function on the order \(0 < 1\), if we mark
both \(0\) and \(1\) in the domain, but only \(1\) in the codomain.
In set theory, this corresponds to the fact that
\(\set{\emptyset,\set{\emptyset}}\) is not a subset of
\(\set{\set{\emptyset}}\).

\begin{lemma}[\flink{Lemma-51}]\label{lem:standard-properties-of-simus}
  For mewos $X$, $Y$, and $Z$, we have the following properties of simulations:
  \begin{enumerate}[(i)]
  \item The underlying function $f$ of a simulation $X \leq Y$ is injective:
    $f\, x_1 = f\, x_2$ implies $x_1 = x_2$.
  \item There is at most one simulation between any two mewos, i.e., $X \leq Y$
    is a proposition.
  \item Simulations are antisymmetric, i.e.
    \[
      X \leq Y \to Y \leq X \to X = Y.
    \]
  \item We have the trivial simulation $X \leq X$ and simulations can be
    composed, i.e.
    \[
      X \leq Y \to Y \leq Z \to X \leq Z.
    \]
  \item $X = Y$ is a proposition, i.e., $\mewo$ is a set.
  \item In the property \eqref{eq:simu-property-for-mewo} of the definition of a
    simulation, the symbol $\exists$ can equivalently be replaced by $\Sigma$.
  \end{enumerate}
\end{lemma}
\begin{proof}
  The arguments are copies of the proofs for type-theoretic ordinals
  (cf.~\cite[Lem~10.3.12, Cor~10.3.13\&15, Lem~10.3.16]{HoTTBook}).
\end{proof}

\begin{definition}[\flink{Definition-52} Initial segment, \({X \down x}\)]%
  \label{def:mewo-initial-segment}
  If $X$ is a mewo and $x : X$, then the \emph{initial segment} $X \down x$ is
  the mewo of elements transitively below $x$, with the canonical inherited
  order. The marked elements are the immediate predecessors \nolinebreak of \nolinebreak\(x\).
\end{definition}

That is, in detail, the carrier of \(X \down x\) is given by
the type ${\Sigma(x' : X).}{(x' \tc x)}$, the order by
$(x_1, s) < (x_2, t) \defeq {(x_1 < x_2)}$, and the marking by 
$\mar (x_1,s) \defeq (x_1 < x)$.

\begin{lemma}[\flink{Lemma-53}]
  The mewo $(X \down x)$ is covered for every $x$.
\end{lemma}
\begin{proof}
  Given $(x_1,s)$ with $s : (x_1 \tc x)$, we wish to show that there exists a
  marked $p$ such that $x_1 \trc p$.  Since we are proving a proposition, we may
  assume that $s$ is a sequence $x_1 < \ldots < x_n < x$, and $x_n$ is marked by
  definition.
\end{proof}

For type-theoretic ordinals, a bounded simulation is a simulation whose domain
is equivalent to an initial segment under a certain element of its codomain.
For mewos, we ensure that the latter property is true by definition. Some
caveats and subtleties are discussed in \cref{subsec:mewo-simulations-failures}
below.

\begin{definition}[\flink{Definition-54} Bounded simulation, \(<\)]%
  \label{def:bounded-sim}
  A \emph{bounded simulation} between mewos $X$ and $Y$ is a pair $(y,e)$, where
  $y : M_Y$ is a marked element in $Y$ and $e : X \simeq (Y \down y)$ an
  equivalence of mewos.  We write $X<Y$ for the type of such pairs.
\end{definition}
It is important that the above definition specifies that $y$ is marked, in line
with our earlier explanation that exactly the marked elements of a mewo
correspond to elements of a $\V$\nobreakdash-set.  We will later
(\cref{cor:bounded-sims-are-unique}) see that the type $X<Y$ is a proposition.
For now, let us observe the following:

\begin{lemma}[\flink{Lemma-55}]\label{lem:<-is-wellfounded}
  The relation $<$ is wellfounded on $\mewo$ and $\cewo$.
\end{lemma}
\begin{proof}
  Since $\cewo \hookrightarrow \mewo$ is order-preserving, it suffices to check
  that $\mewo$ is wellfounded.  Thus, we need to show that every mewo $X$ is
  accessible, i.e., that all its predecessors are accessible.  By definition,
  every predecessor is of the form $X \down x_0$ for some marked $x_0 : X$.

  Exploiting that the order on $X$ itself is wellfounded, we show by transfinite
  induction on $x$ the more general statement that every $X \down x$ is
  accessible, no matter whether $x$ is marked.  Thus, assume that, for all
  $z<x$, we have that $X \down z$ is accessible.  We need to prove that all
  predecessors of $X \down x$, i.e., all $(X \down x) \down (x_1,p)$, are
  accessible.  An adaption of \cref{iterated-initial-segments} for mewos shows
  that this mewo is equal to $X \down x_1$, which is accessible by the induction
  hypothesis.
\end{proof}

\subsection{Subtleties caused by markings}\label{subsec:mewo-simulations-failures}

Observing how bounded simulations interact with other (possibly bounded)
simulations reveals the complete change of view we are forced to make when
generalizing from ordinals to mewos.  This is indeed intended since we claim
(and prove) that $\leq$ and $<$ correspond to $\subseteq$ and $\in$, and for
arbitrary sets, the latter relations fail to have many properties that one might
associate with the former relations.

The first point is that $X < Y$ generally does not imply $X \leq Y$.  A bounded
simulation $X < Y$ gives rise to a function $f : X \to Y$ via the composition of
the function underlying $e : X = (Y \down y)$ and the first projection
$(Y \down y) \to Y$.  However, the first projection is in general not a
simulation as it may not preserve markings.
A counter-example is the covered mewo $\circ \leftarrow \bullet$, i.e., the mewo
with two comparable elements, the larger of which is marked (denoted by
$\bullet$), while the smaller is not (denoted by $\circ$).  Since
$(\circ \leftarrow \bullet) \down \bullet$ is, by definition, simply $\bullet$,
there is a bounded simulation from $\bullet$ to $\circ \leftarrow \bullet$.
However, there is no simulation as the marking is not preserved. The crux here
is that the operation $\down$ changes the marking.
The translation to the language of sets is that $\{\emptyset\}$ is an element,
but not a subset, of $\{\{\emptyset\}\}$.

Secondly, the \(<\) relation on mewos is not transitive. The principle that
``an initial segment of an initial segment is an initial segment'' does not
hold.  A simple counter-example is the empty mewo $\emptyset$ together with
$\bullet$ and $\circ \leftarrow \bullet$.  We have bounded simulations
$\emptyset < \bullet < (\circ \leftarrow \bullet)$, but no bounded simulation
$\emptyset < (\circ \leftarrow \bullet)$.  In this case, the translation is that
$\emptyset$ is an element of $\{\emptyset\}$, which itself is an element of
$\{\{\emptyset\}\}$; the latter however does not have $\emptyset$ as an element.

For technical reasons, it is occasionally useful to use mewos that ensure that
the discussed properties do hold. This can be achieved by changing the marking
to the trivial one:
\begin{definition}[\flink{Definition-56} Trivializing the marking,
  \(\overline X\)]\label{def:mark-all}
  If $X$ is a mewo, we write $\overline X$ for the mewo that has the same
  carrier and order as $X$, but where every element is marked.
\end{definition}

In the language of $\V$-sets, $\overline X$ is the union of all the sets
represented by elements (of elements of elements \ldots) of~$X$.  Note that
$\overline X$ is still not transitive and thus not a type-theoretic
ordinal. Nevertheless, this operation allows us to recover several important
properties of type-theoretic ordinals:

\begin{lemma}[\flink{Lemma-57}]
  For given mewos $X$, $Y$, and $Z$, we have:
  \begin{enumerate}[(i)]
  \item for every $x:X$, the first projection $(X \down x) \to \overline X$ is a
    simulation;
  \item $X < Y \to X \leq \overline Y$;
  \item $X < Y \to Y < Z \to X < \overline Z$.
  \end{enumerate}
\end{lemma}
\begin{proof}
  As the conditions involving markings now are vacuously true, the arguments for
  type-theoretic ordinals apply.
\end{proof}

As a demonstration of how this is useful, we can show the following technical lemma:

\begin{lemma}[\flink{Lemma-58}]\label{lem:down-is-injective}
  Given a mewo $X$, the function $x \mapsto X \down x$ is injective:
  $(X \down x_1) = (X \down x_2)$ implies $x_1 = x_2$.
\end{lemma}
\begin{proof}
  We show that $f: (X \down x_1) \leq (X \down x_2)$ implies that any
  predecessor of $x_1$ is also a predecessor of $x_2$; extensionality of $X$
  then gives the claimed injectivity.
  To do this, let us consider the following diagram:
  \[
      \begin{tikzpicture}[x=3cm,y=-1cm,baseline=(current bounding box.center)]
        \node (X1) at (0,0) {$X \down x_1$};
        \node (X) at (.5,1) {$\overline X$};
        \node (X2) at (1,0) {$X \down x_2$};

        \draw[->] (X1) to node[above] {$f$} (X2);
        \draw[->] (X1) to node[below left] {$\mathsf{fst}$} (X);
        \draw[->] (X2) to node[below right] {$\mathsf{fst}$} (X);
      \end{tikzpicture}
  \]
  All maps are simulations and, by uniqueness of simulations
  (\cref{lem:standard-properties-of-simus}), the diagram necessarily commutes.
  Given a predecessor ${x < x_1}$, it is marked in $X \down x_1$ by
  construction, and since $f$ preserves markings, $f \, x$ is marked as well,
  i.e., we have $f \, x < x_2$.  But since the diagram commutes, we have
  $f\, x = x$ as elements of $X$.
\end{proof}

A consequence is that bounded simulations are unique:
\begin{corollary}[\flink{Corollary-59}]\label{cor:bounded-sims-are-unique}
  For mewos $X$ and $Y$, the type $X < Y$ of bounded simulations is a
  proposition.
\end{corollary}
\begin{proof}
  By definition, $X < Y \equiv \Sigma(y : Y). (X = (Y \down y))$.  Assume
  $(y, p), (y', q) : X < Y$. By the above lemma, we then have $y = y'$ since
  $(Y \down y) = X = (Y \down y')$, and $p = q$ since $\mewo$ is a set. Hence
  $(y, p) = (y', q)$, as desired.
\end{proof}

\subsection{Simulations and coverings}

As we have seen, bounded simulations and simulations are tricky to compare.  The
first step towards improving this situation is to characterize a simulation via
initial segments:

\begin{lemma}[\flink{Lemma-60}]\label{lem:sim-equal-initial-segments-new}
  Let $X$ and $Y$ be mewos. Further, let $f : X \to Y$ be a function between the
  carriers that preserves markings, i.e., such that $\mar(x) \to \mar(f \, x)$.  The
  following are equivalent:
  \begin{enumerate}[(i)]
  \item\label{item:f-simu-vs-partial-simu-new-1} $f$ is a simulation.
  \item\label{item:f-simu-vs-partial-simu-new-2} for all $x:X$, we have
    $(X \down x) = (Y \down (f \, x))$.
  \end{enumerate}
\end{lemma}
\begin{proof}
 $(\ref{item:f-simu-vs-partial-simu-new-1}) \Rightarrow (\ref{item:f-simu-vs-partial-simu-new-2})$:
 An equality of mewos is a surjective simulation that preserves and reflects the
 markings.  The simulation $f$ is monotone and thus can be restricted to a
 simulation $\bar f : X \down x \leq Y \down (f \, x)$.  Monotonicity of $f$
 guarantees that markings are preserved, while the initial segment property
 ensures that markings are reflected.  Finally, by induction on the number of
 steps, the initial segment property for $<$ can be extended to $\tc$; hence
 every $y$ in $Y \down (f \, x)$ has a preimage.

 $(\ref{item:f-simu-vs-partial-simu-new-2}) \Rightarrow
 (\ref{item:f-simu-vs-partial-simu-new-1})$: Assume that, for every $x$, we have
 an equality $e_x : X \down x = Y \down (f \, x)$ of mewos.

 It is a standard result that the transitive closure of a wellfounded relation
 is wellfounded.  Using this we show, by transfinite induction on $x$, that $f$
 is a simulation at point $x$:
 \begin{itemize}
 \item for $x_1 < x$ we have $f \, x_1 < f \, x$;
 \item for $y_1 < f \, x$, there is $x_1 < x$ such that $f \, x_1 = y_1$.
 \end{itemize}
 The induction hypothesis states that $f$ is a simulation at every point $z$
 with $z \tc x$ or, in other words, that the composition
 $(X \down x) \xrightarrow{\mathsf{fst}} \overline X \xrightarrow f \overline Y$
 is a simulation (cf.~\cref{def:mark-all}).  Therefore, the diagram
\begin{equation}\label{eq:diagram-using-unique-simulations}
    \begin{tikzpicture}[x=3cm,y=-1cm,baseline=(current bounding box.center)]
      \node (X1) at (0,0) {$X \down x$};
      \node (X) at (0,1) {$\overline X$};
      \node (Y1) at (1,0) {$Y \down (f \, x)$};
      \node (Y) at (1,1) {$\overline Y$};

      \draw[->] (X1) to node {} (X);
      \draw[->] (X1) to node[above] {$e_x$} (Y1);
      \draw[->] (X) to node[above] {$f$} (Y);
      \draw[->] (Y1) to node {} (Y);
    \end{tikzpicture}
\end{equation}
commutes by uniqueness of simulations (\cref{lem:standard-properties-of-simus}).
We can now easily check that $f$ is a simulation at point $x$.  First, ${z < x}$
means that $z$ is a marked element in $X \down x$, thus $e_x \, z$ is marked in
$Y \down (f \, x)$, translating to $\mathsf{fst} (e_x \, z) < f \, x$, and
commutativity of \eqref{eq:diagram-using-unique-simulations} implies
$\mathsf{fst} (e_x \, z) = f \, z$.  Second, let $y_1 < f\, x$ be given. This
means that $y_1$ is marked in $Y \down (f \, x)$ and we get the marked $x_1$ as
the unique preimage of $y_1$ under the equivalence $e_x$.
\end{proof}

While we have seen in \cref{subsec:mewo-simulations-failures} that $<$ is not
transitive and does not necessarily imply $\leq$, we now get the following
familiar property:
\begin{corollary}[\flink{Corollary-61}]
  For mewos $X$, $Y$ and $Z$, we have
  \[
    X < Y \to Y \leq Z \to X < Z.
  \]
\end{corollary}
\begin{proof}
  We have $X = (Y \down y)$ by assumption and $(Y \down y) = (Z \down f\, y)$ by
  \cref{lem:sim-equal-initial-segments-new}.
\end{proof}

One may view \cref{lem:sim-equal-initial-segments-new} as stating that a
function is a simulation if and only if it behaves like a simulation pointwise
(or \emph{locally}).  We now consider such functions that are only defined on
the marked elements:

\begin{definition}[\flink{Definition-62} Partial simulation, \({\simbelow}\)]
  A \emph{partial simulation} between mewos \(X\) and \(Y\) is a function
  ${f : \Mar_X \to \Mar_Y}$ that preserves initial segments,
  \[
    \presinitial {} f \; \defeq \; \forall (x:\Mar_X). (X \down x) = (Y \down f
    \, x),
  \]
  and we write
  \( X \simbelow Y \; \defeq \; \Sigma (f : \Mar_X \to \Mar_Y). \presinitial {}
  f.  \)
\end{definition}

A convenient alternate representation is the following:
\begin{lemma}[\flink{Lemma-63}]\label{cor:sim-vs-ptwise}
  The type of partial simulations $X \simbelow Y$ is equivalent to the type
  \[
    \forall (x : \Mar_X). \exists (y : \Mar_Y). X \down x = Y \down y.
  \]
  and hence a proposition.
\end{lemma}
\begin{proof}
  We can calculate
  \begin{alignat*}{9}
    X \simbelow Y
    & \equiv & \quad & \Sigma (f : \Mar_X \to \Mar_Y). \forall x. (X \down x) = (Y \down f \, x)   \\
    & \simeq & \quad & \Pi(x : \Mar_X). \Sigma(y : \Mar_X). (X \down x) = (Y \down y) \\
    & \simeq & \quad & \hspace{1pt}\forall\hspace{1pt}(x : \Mar_X). \hspace{1pt}\exists\hspace{1pt}(y : \Mar_X). (X \down x) = (Y \down y),
  \end{alignat*}
  where the first step is the definition of $\simbelow$, the second is the
  ``untruncated axiom of choice''~\cite[Thm~2.15.7]{HoTTBook}, and the last step
  uses that $(Y \down y_1) = (Y \down y_2)$ implies ${y_1 = y_2}$ by
  \cref{lem:down-is-injective}, which means that $\Sigma$ and $\exists$ are
  equivalent.
\end{proof}

A mewo can have the property that its marking alone already fully determines how
it maps into other mewos.  The notions introduced above allow us to make this
precise:
\begin{definition}[\flink{Definition-64} Principality]\label{def:principality}
  The marking of a mewo \(X\) is \emph{principal} if, for all mewos $Y$, the
  canonical restriction map $(X \leq Y) \to (X \simbelow Y)$ given by
  \cref{lem:sim-equal-initial-segments-new} is an equivalence.
\end{definition}

In other words, for any chosen codomain $Y$, the marking of $X$ is principal if
a (necessarily unique) partial simulation out of $X$ already determines a
(necessarily unique) simulation out of $X$.  However, being principal is
actually simply a ``relative'' description of the ``absolute'' property of being
covering:

\begin{lemma}[\flink{Lemma-65}]\label{lem:cover-iff-principal}
  A marking covers if and only if it is principal.
\end{lemma}
\begin{proof}
  Let $\mar$ be a marking on a mewo $X$.
  \paragraph{\textbf{covers} $\Rightarrow$ \textbf{principal}}
  Assume we have a partial simulation $f : X \simbelow Y$.  For a given $x : X$,
  we need to find a (necessarily unique) $y : Y$ such that
  $X \down x = Y \down y$.  By the covering property, there exists
  $x_0 : \Mar_X$ with $p : x \trc x_0$.  By analyzing $p$, we get either
  $x = x_0$, in which case the goal is given by the partial simulation, or
  $x \tc x_0$.  In the latter case, we get
  $e : (X \down x_0) = (Y \down f \, x_0)$ from the partial simulation.
  Applying the function underlying $e$ on $x$, we generate an element $y:Y$ that
  satisfies the required property.  If $x$ is marked, then the (unique) $y$ that
  we find is necessarily equal to the one given by the partial simulation, which
  is marked by assumption.

\paragraph{\textbf{principal} $\Rightarrow$ \textbf{covers}}
Assume $\mar$ is principal. Let $\widehat X$ be the mewo of all elements covered
by $\Mar_X$, defined as
\[
  \widehat X \; \defeq \; \Sigma(x:X). \exists (x_0 : \Mar_X). (x \trc x_0),
\]
with order and marking inherited from $X$.  We have ${\widehat X \leq X}$ by
projection.  We also have $X \simbelow \widehat X$ by definition and thus
$X \leq \widehat X$ by principality, meaning that the two mewos are equal by
antisymmetry. In other words, $\mar$ covers all of $X$.
\end{proof}

We have seen in \cref{lem:<-is-wellfounded} that $<$ is wellfounded on $\mewo$
and $\cewo$.  The observation that principality and covering coincide allows us
to show that, in the latter case, the order is also extensional:

\begin{theorem}[\flink{Theorem-66}]\label{thm:covered-ewo-extensional}
  The structure $(\cewo,<)$ is an extensional wellfounded order.
\end{theorem}
\begin{proof}
  Wellfoundedness has been established in \cref{lem:<-is-wellfounded}.
  Extensionality follows antisymmetry
  (\cref{lem:standard-properties-of-simus}) as soon as we can show that
  \begin{equation}\label{eq:half-ext} \forall (Z : \cewo). (Z < X) \to (Z < Y)
  \end{equation}
  implies $X \leq Y$.  Thus, let us prove this property.

  Let $X$ and $Y$ be covered mewos.  By principality and
  \cref{cor:sim-vs-ptwise}, we need to show that for every \(x : \Mar_X\) there
  exists some \(y : \Mar_Y\) with \(X \down x = Y \down y\).  By definition, the
  predecessors of $X$ are exactly the mewos of the form $X \down x$ for marked
  $x$, so that this formula is equivalent to the assumption~\eqref{eq:half-ext}.
\end{proof}

In contrast, the relation $<$ is clearly not extensional on $\mewo$, as there are
many different mewos without predecessors, namely exactly those with completely
empty markings.

\subsection{Constructions on mewos}

Recall the rank function $\Psi : \V \to \Ord$ from \cref{def:Psi}.  Since
different $\V$-sets can have the same rank, $\Psi$ is not injective and thus
certainly not a simulation.  We have seen that we can turn it into a simulation
by restricting its domain to $\Vord$.  This is of course not sufficient anymore
for our current goal of characterizing \emph{all} of $\V$; instead, we extend
the codomain from $\Ord$ to $\mewo$.  Doing this requires us to generalize the
operations on $\Ord$ that we used to construct the rank function.
In \cref{sum-of-ordinals}, we recalled the addition of type-theoretic ordinals.
While it would be possible to phrase this definition in full generality for
mewos, we restrict ourselves for simplicity to the case of interest (the
successor), which already contains the crucial ideas.

There is however an important difference.  The successor operation for
type-theoretic ordinals, if translated to and written in the notation of set
theory, maps a set $S$ to $S \cup \{S\}$.
This is of course required in order not to leave the realm of \emph{transitive}
sets (and orders).  For mewos, we need to slightly refine the function so that
it corresponds to the (non-transitive) \emph{singleton operation}
$S \mapsto \{S\}$.

\begin{definition}[\flink{Definition-67} Singleton, \(\sing{X}\)]
  For a given mewo $X$, we define the \emph{singleton order} $\sing{X}$ to be
  the marked order with carrier $X + \mathbf 1$ and the order given as follows:
  \begin{itemize}
  \item $(\inl x < \inl y)$ if and only if $x < y$;
  \item $(\inl x < \inr \star)$ if and only if $\mar(x)$;
  \item $(\inr \star < z)$ false for all $z$.
  \end{itemize}
  Finally, we mark the single point $\inr \star$.
\end{definition}
It is worth pointing out how this \emph{almost} generalizes the successor
operation of a type-theoretic ordinal.  Since such an ordinal is a completely
marked (and transitive) mewo, the second clause above matches exactly the sum
operation given in \cref{sum-of-ordinals} when the second summand is $\One$.
However, a faithful generalization of \cref{sum-of-ordinals} would in the end
mark not only $\inr \star$, but also all elements that were marked in $X$.

Another critical point to note is that, for an arbitrary mewo $X$, the
singleton $\sing{X}$ need not be a mewo.  As an example, consider the mewo
$\circ$ with exactly one element, which is unmarked. If we now take its
singleton, neither this existing element nor the newly added element has any
predecessors. Since they are not equal, extensionality is missing.
The obstacle in this example is that the original marking is insufficient.
Fortunately, if we start with a \emph{covered} mewo, the successor is not only
extensional but also covered again:
\begin{lemma}[\flink{Lemma-68}]
  If $X$ is a covered mewo, then so is $\sing{X}$.
\end{lemma}
\begin{proof}
  Wellfoundedness is immediate. Regarding extensionality, the interesting case
  is comparing an element of the form $\inl x$ with $\inr \star$. It suffices to
  show that their predecessors are not the same. To do so, observe that $x$ is
  covered in $X$, i.e., there exists $x_0$ with $x \trc x_0$.  By construction,
  $x_0$ is a predecessor of $\inr \star$, while wellfoundedness ensures that it
  cannot possibly be a predecessor of $\inl x$.
  Coveredness: The element $\inr \star$ is marked and thus trivially covered.
  To see that an arbitrary $\inl x$ is covered, note that there exists a marked
  $x_0$ with $x \trc x_0$ in $X$. By construction, we have
  $\inl x_0 < \inr \star$, implying that $\inl x$ is covered.
\end{proof}

The second important construction that we discussed for type-theoretic ordinals
is computing suprema (\cref{sup-of-ordinals}).  In the case of mewos, the better
intuition is to think of \emph{unions}, although the universal property of the
supremum is satisfied too, as we will see shortly in
\cref{lem:mewo-sup-is-sup}.

\begin{definition}[\flink{Definition-69} Union of mewos, \(\bigcup F\)]%
  \label{def:union-of-ewos}
  The \emph{union} $\bigcup F$ of a family of mewos $F : A \to \mewo$ is defined
  as follows:
  \begin{itemize}
  \item The carrier is $\Sigma(a:A).F \, a$ quotiented by $\approx$, where we
    define $(a,x) \approx (b,y)$ to be \( (Fa \down x) \simeq (Fb \down y) \) as
    (covered) mewos;
  \item and $[a,x] < [b,y]$ is defined as
    \( {(F \, a \down x) < (F \, b \down y)}.  \)
  \end{itemize}
  We mark $s : \bigcup F$ if and only if there exist $a_0 : A$ and
  $x_0 : F\,a_0$ with $s = [a_0,x_0]$ such that $x_0$ is marked in $F \, a_0$.
\end{definition}
\begin{remark}
  The explanation given in \cref{rem:equivalence-instead-of-equality} applies.
  A~priori, the type $(Fa \down x) = (Fb \down y)$ is too large as it lives in a
  higher universe than the mewos in consideration, which is why we use $\simeq$
  in the definition above. The issue is also extensively discussed in
  \cref{subsec:revisit-rank}.
\end{remark}

Continuing the observation that mewos act as sets and simultaneously generalize ordinals, we note that the union is also a supremum:

\begin{lemma}[\flink{Lemma-71}]\label{lem:mewo-sup-is-sup}
  $\bigcup F$ is the least upper bound of all $F(a)$.
\end{lemma}
\begin{proof}
  $F \, a \leq \bigcup F$ is easy to check.  Assume now that we have
  $F \, a \leq Y$ for every $a$; we want to prove $\bigcup F \leq Y$.  By a
  calculation analogous to the one in \cref{cor:sim-vs-ptwise}, this goal means
  we need to show that, for any $z : \bigcup F$, there exists a $y: Y$ such that
  $(\bigcup F \downarrow^+ z) = (Y \downarrow^+ y)$ and $\mar(z) \to \mar (y)$.
  This follows by induction on $z$, using the uniqueness of $y$ and the assumption
  for the marking condition.
\end{proof}

\begin{lemma}[\flink{Lemma-72}]
  If $F$ is a family of covered mewos, then $\bigcup F$ is covered.
\end{lemma}
\begin{proof}
  Let $[a,x]$ be an element of $\bigcup F$; we want to show that $[a,x]$ is
  covered.  By assumption, $x$ is covered in $F \, a$ by some $x_0$.  Since the
  operation $(F \, a \down \_)$ preserves $<$, it also preserves $\trc$ and we
  get $F \, a \down x \trc F \, a \down x_0$, giving $[a,x] \trc [a,x_0]$ as
  required.
\end{proof}

\begin{remark}
  Note that, in the situation of \cref{def:union-of-ewos}, we can have
  $(a,x) \approx (b,y)$ such that $x$ is marked while $y$ is not.  The simplest
  example when this happens is the union of the mewos
  $\bullet \leftarrow \bullet$ and $\circ \leftarrow \bullet$
  (cf.~\cref{subsec:mewo-simulations-failures} for the notation), in
  set-theoretic notation corresponding to the union of
  $\{\{\emptyset\},\emptyset\}$ and $\{\{\emptyset\}\}$.  Therefore, it is
  important to phrase the marking condition in \cref{def:union-of-ewos} using an
  \emph{exists} instead of \emph{forall}.
\end{remark}

\subsection{\texorpdfstring{$\V$}{V}-sets and covered mewos coincide}

We are ready to prove our second main theorem, and complete the
square~\eqref{eq:square-completed} by showing that $\V$ and $\cewo$ coincide.
We have seen that the relation $\in$ on $\V$ is wellfounded and extensional.  By
marking everything, $\V$ is therefore a (large) mewo.  Similarly, $\cewo$ itself
is a (large) mewo, using \cref{thm:covered-ewo-extensional} and total marking.
To show that they are equal as such, we construct simulations between them.

\begin{lemma}[\flink{Lemma-74}]\label{lem:V-sim-cewo}
  We have a simulation $\V \leq \cewo$.
\end{lemma}
\begin{proof}
  We define the function $\Psi : \V \to \cewo$ underlying the simulation by
  induction on the input by defining
  \[
    \Psi(\Vset{A}{f}) \; \defeq \; \bigcup_{a : A} \left(\sing{\Psi(f \,
        a)}\right).
  \]
  We need to verify that extensionally equal representatives are mapped to equal
  mewos, which follows from \cref{lem:mewo-sup-is-sup}.

  The following observation is helpful to see that $\Psi$ is a simulation: the
  predecessors (i.e., elements) of $\Vset{A}{f}$ are exactly the elements of the
  form $f(a_0)$ for $a_0 : A$, and similarly, via a quick calculation, the
  predecessors of $\bigcup_{a : A} \left(\sing{\Psi(f \, a)}\right)$ are of the
  form $\Psi(f\, a_0)$.

  Regarding monotonicity, assume we have elements ${v_1 \in v_2}$ in $\V$.  By
  induction on $v_2$, we may assume that it is of the form $\Vset{A}{f}$, and
  its predecessor $v_1$ is therefore of the form $f(a_0)$.  As we have just
  seen, we then have the desired $\Psi(f \, a_0) < \Psi(\Vset{A}{f})$.
  Regarding the second property, we proceed similarly.  Given any
  $y \in \Psi(\Vset{A}{f})$, we know that $y$ is of the form $\Psi(f \, a_0)$,
  and hence we have $f(a_0) \in \Vset{A}{f}$ as required.
\end{proof}

\begin{lemma}[\flink{Lemma-75}]\label{lem:cwo-sim-V}
  We have a simulation $\cewo \leq \V$.
\end{lemma}
\begin{proof}
We define the function $\Phi : \cewo \to \V$ by
\[
\Phi(X) \; \defeq \; \Vset{\Mar_X}{\lambdadot{x_0}{\Phi(X \down x_0)}}.
\]
The predecessors of $X$ are of the form $X \down x_0$ for $x_0 : \Mar_X$, while
the elements of $\Phi(X)$ are $\Phi(X \down x_0)$ for $x_0 : \Mar_X$.
Therefore, the simulation properties for $\Phi$ follow analogously to how we
derived them in the proof of \cref{lem:V-sim-cewo}.
\end{proof}

By \cref{lem:V-sim-cewo,lem:cwo-sim-V}, and antisymmetry, we get:
\begin{theorem}[\flink{Theorem-76}]\label{thm:second-main-result}
  The structures $(\V,\in)$ and $(\cewo,<)$ are equal as covered mewos.\qedNoProof
\end{theorem}

\section{Conclusion}
Working in homotopy type theory, we have shown that the set-theoretic ordinals in
\(\V\) coincide with the type-theoretic ordinals.
Moreover, by generalizing from type-theoretic ordinals to covered mewos, we
have captured \emph{all} sets in \(\V\).

A natural question is whether similar results can be obtained by working inside
set theory instead. E.g., we expect the type-theoretic ordinals in the cubical
sets model~\cite{BezemCoquandHuber2014} of homotopy type theory to coincide with
the set-theoretic ordinals, using the Mostowski collapse
lemma~\cite{Mostowski1949}.
Another, orthogonal question is whether the presentation of \(\V\) as the type
of covered mewos can shed any light on the open problem~\cite[below
Cor~10.5.9]{HoTTBook} of whether \(\V\) satisfies the \emph{strong collection}
and \emph{subset collection} axioms of Constructive ZF set theory.
Moreover, it would be interesting to study how other, different notions of
constructive ordinals, such as Taylor's \emph{plumb ordinals}~\cite{Taylor1996},
behave in a type-theoretic setting.

\ifCLASSOPTIONcompsoc
  \section*{Acknowledgments}
\else
  \section*{Acknowledgment}
\fi

We would like to thank Andreas Abel, who asked us how the type-theoretic
ordinals and the ordinals in Aczel's interpretation of set theory in type theory
might be related.
We are also grateful to Mart\'in Escard\'o for discussions on ordinals and the
ability to build on his Agda development.
Finally, we are thankful to Paul Levy for several valuable suggestions.

Funding: This work was supported by The Royal Society (grant reference URF\textbackslash{}R1\textbackslash{}191055) and the UK National Physical Laboratory Measurement Fellowship project ``Dependent types for trustworthy tools''.

\IEEEtriggeratref{3}


\bibliographystyle{IEEEtran}
\bibliography{IEEEabrv,references}

\end{document}